\documentclass[10pt, twocolumn]{IEEEtran}

\usepackage{ifpdf} 
\usepackage{array} 

\usepackage{amsmath}
\usepackage{amsthm} 

\theoremstyle{remark}
\newtheorem{theorem}{Theorem}

\newtheorem{definition}{Definition}
\newtheorem{proposition}{Proposition}
\newtheorem{remark}{Remark}

\theoremstyle{remark}
\usepackage{amssymb}    %
\usepackage{amsfonts}
\usepackage[english]{babel}
\usepackage{cite}
\usepackage{multirow}
\usepackage{multirow}
\usepackage{makecell}
\usepackage{stfloats}   
\usepackage{algorithm}   
\usepackage{algorithmic} 
\usepackage{graphicx}
\usepackage{caption}
\usepackage{epsfig}
\usepackage{subfigure} 
\usepackage{cite}
\usepackage{tablefootnote}
\usepackage{diagbox}   

\usepackage{setspace} 
\usepackage{threeparttable}

\usepackage{url}

\usepackage{footnote}
\usepackage[textsize=tiny]{todonotes}
\usepackage{marginnote}
\usepackage{pifont}
\usepackage[perpage]{footmisc}

\usepackage{booktabs}  

\allowdisplaybreaks  

 \usepackage[draft]{changes} 

\usepackage{color}

\begin{document}
\title{On-Time Communications Over Fading Channels}
%
%
%

\author{Yan~Li, Yunquan Dong,~\IEEEmembership{Member,~IEEE},~and Byonghyo~Shim,~\IEEEmembership{Senior Member,~IEEE}
\thanks{Yan Li and Y. Dong are with the School of Electronic and Information Engineering,  Nanjing University of Information Science and Technology, Nanjing 210044, China (e-mail: \{yanli, yunquandong\}@nuist.edu.cn).
        }
\thanks{B. Shim is with the Department of Electrical and Computer Engineering, Seoul National University, Seoul 151-744, Korea (e-mail: bshim@snu.ac.kr).}
       }

\maketitle

\begin{abstract}
We consider the on-time transmissions of a sequence of packets over a fading channel.
    Different from traditional in-time communications, we investigate how many packets can be received $\delta$-on-time, meaning that the packet is received with a deviation no larger than $\delta$ slots.
    In this framework, we first derive the on-time reception rate of the random transmissions over the fading channel when no controlling is used.
To improve the on-time reception rate, we further propose to schedule the transmissions by delaying, dropping, or repeating the packets.
    Specifically, we model the scheduling over the fading channel as a Markov decision process (MDP) and then obtain the optimal scheduling policy using an efficient iterative algorithm.
For a given sequence of packet transmissions, we analyze the on-time reception rate for the random transmissions and the optimal scheduling.
    Our analytical and simulation results show that the on-time reception rate of random transmissions decreases (to zero) with the sequence length.
By using the optimal packet scheduling, the on-time reception rate converges to a much larger constant.
    Moreover, we show that the on-time reception rate increases if the target reception interval and/or the deviation tolerance $\delta$ is increased, or the randomness of the fading channel is reduced.
%
\end{abstract}


\IEEEpeerreviewmaketitle

\section{Introduction}

\IEEEPARstart{W}{ith} the rapid development of the industrial Internet technology, 5G communications, and Internet-of-Things (IoT) technology, billions or even trillions of smart devices will be connected to the internet to enable  efficient interactions between the physical world and its digital counterpart.
    On the one hand, surge of industrial machine-type communications furthered the possibilities for new applications in various industry areas.
On the other hand, applications like industrial sensing and controlling, remote surgery, and automatic driving, require a very low latency (e.g., end-to-end delay being smaller than 10 ms) and a very small jitter (approximately several milliseconds) \cite{Sisnni-IIoT.2018, Wollschlaeger.M.2017, Ge.Xiaohu.2019, shim.2018}.
    For example,  communications between the sensor, actuators, and controller of an industrial Internet should be completed  on-time with a deterministic delay between 1 and 10 ms \cite{liuyunjie-determine.2019};
the braking/steering commands and advanced driver assistance systems (ADAS) type data need to be delivered to/from the actuators/sensors with a deterministic latency being less than 1 ms \cite{Ge.Xiaohu.2019}.
    Therefore, how to deliver information in-time, or even on-time, has become one of the the biggest challenge of modern wire-line and wireless communications.

Owing to the high reliability of cable (or optical fiber) communications, wire-line networks were the first choice in deterministic-latency information deliveries.
    Based on the widely used Ethernet, IEEE 802.1 working group has developed a series of time sensitive networking (TSN) standards for time sensitive applications \cite{Time-sensitive.networking.2018}.
By scheduling traffics with timed transmission gates, filtering traffics based on priorities, and forwarding traffics with repeating circles, TSN networks can deliver the traffics with deterministic delays.
    Thus, TSN has become the basis of time sensitive applications like industrial automation and automotive driving.
In an in-vehicle TSN network, for example, the communications between the vehicle control unit (VCU)  and the cameras, radars, lidars, and the positioning module, can be guaranteed to be deterministic and timely (less than 1 ms) \cite{Ge.Xiaohu.2019}.
    Furthermore, the inter-vehicle communications can be realized by 5G ultra reliable low latency communication (URLLC), as shown in Fig. \ref{fig:1_tsn_urllc}.
By using techniques such as mini-slot scheduling, multi-access computing, and downlink preemption scheduling, URLLC achieves a round-trip  air-interface delay of 2.7 ms (almost deterministic) \cite{shim.2018}.
    Although the combination of TSN and URLLC can offer a satisfying solution for vehicular communications \cite{JSAC2021}, one important question remains:
Is it possible to replace the wire-lines of TSNs with wireless channels?

\begin{figure}[!t]
  \centering
  \includegraphics[width=3.3in]{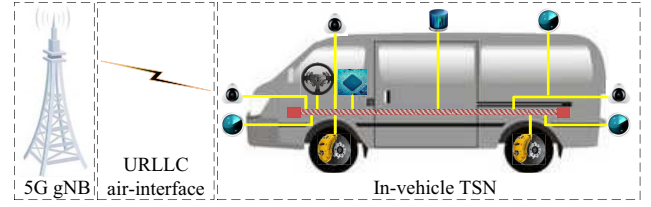}\\
  \caption{An exemplary 5G URLLC and TSN based vehicular network, in which sensors are connected to the VCU through wire-line TSN and the VCU is connected with other vehicles and the internet through 5G URLLC.}
  \label{fig:1_tsn_urllc}
\end{figure}

\subsection{Overview and Main Contributions}

In this paper, we will study how wireless channels can support on-time communications.
    Instead of delivering the packets over fading channels using the best effort principle, we focus more on how many packets can be delivered to the receiver at the predefined epochs (with no delay or ahead of time) or in reception ranges with small deviations (no larger than $\delta$).
Note that in case the packet can be delivered on-time, corresponding delays and system status become more predictable, thereby supporting numerous time-sensitive applications in 5G and 6G.

    First of all, we note that due to the randomness and time-varying property of wireless channels, the transmission delays are also random.
Therefore,  to ensure 100 percent deterministic transmissions directly is not possible for the fading channels.
    For this reason, we propose a metric of on-timeliness called $\delta$-on-time reception rate, which is defined as the proportion of received packets whose deviation is no larger than $\delta$.
Second, we derive the $\delta$-on-time reception rate of the random transmissions over fading channels and show that it goes to zero when a length of the packet sequence goes to infinity.
    Specifically,  we explicitly calculate the probability for each packet to be received $\delta$-on-time.
     We also show that the average number of packets received $\delta$-on-time is equivalent to the sum of these probabilities, with which the $\delta$-on-time reception rate can be obtained.
Third, we propose a scheme improving the on-time reception rate by optimally delaying, dropping, or repeating the packets.
    By modeling the packet scheduling problem as a Markov decision process (MDP), we solve the optimal control problem using a simple iterative algorithm.
We also analyze the $\delta$-on-time reception rates of the random transmissions and the optimally controlled transmissions, from which we validate the effectiveness of the proposed algorithm.
    The contributions of the paper are summarized as follows.
\begin{itemize}
\item We propose a method of evaluating the on-time performance of communications in terms of strictly on-time, $\delta$-on-time, and the on-time reception range.
\item We explicitly derive the on-time reception rate of the random transmissions over fading channels and show that the on-time reception rate decreases monotonically (to zero) with the length of the packet sequence.
\item We improve the on-time reception rate over fading channels by delaying, dropping, and repeating packets.
    We also solve the optimal packet scheduling policy through an MDP formulation.
        We demonstrate from simulation that the obtained scheduling policy matches with the theoretical results with negligible error, which show that by using the optimal packet scheduling, the on-time reception rate converges to a constant and significantly outperforms the random transmission scheme.
\end{itemize}

\subsection{Related Works}
The time-sensitive communications and networks have received much attention in recent years, among which the wired TSN \cite{Time-sensitive.networking.2018} and wireless 5G deterministic network \cite{5g-Dnet-2019} are the most representative works.
     First, TSN follows the standard Ethernet protocol system and reduces physical- and link-layer delays by IEEE 802. 1AS clock synchronization, IEEE 802. 1Qcc flow reservation, and IEEE 802. 1Qch cyclic queuing.
Most academic researches in this area focus on the scheduling of messages.
    For example, a computational efficient solution to the fully deterministic 802.1Qbv scheduler was presented in \cite{S.S.Craciunas.2016};
a bandwidth-efficient TSN scheduler was investigated through a size based queueing method in \cite{K.A.Hares.P.2017};
    an asynchronous traffic scheduling algorithm which achieves both low delay and low implementation complexity was proposed in \cite{J.Specht.2016};
and an online scheduling approach was proposed to deal with the dynamic virtual machine migrations in multicast TSN networks in \cite{online-schdl2020}.
In addition, the authors proposed a simple hardware enhancement of switches to increase the schedulability and throughput of time-triggered traffics in \cite{throughput2020}.
    Routing is also an important part of TSN networks, for which an ILP-based scheduling and degree of conflict aware multipath routings scheme was proposed in \cite{routing2020}  while  a joint routing-scheduling optimization for time-triggered Ethernet networks was investigated in \cite{Schweissguth.E.2017}.
Second, URLLC is one of the three major scenarios of the 5G mobile cellular systems \cite{shim.2018}.
      Since 5G URLLC aims at transmitting packets with ultra low delay (2$\sim$12 ms) and ultra high reliability (99.999\%), it is possible to support some dedicated networks providing predictable and deterministic services, which are referred to as the 5G deterministic networking (5GDN or 5G DeNet) \cite{5g-Dnet-2019, 5gdnwhite2020}.
As was reviewed in \cite{5gdnwhite2020}, 5GDN has great opportunities to converge with applications including real-time monitoring, remote controlling, material management, massive access, and product life-cycle management, to name just a few.

Moreover, the age of information (AoI) theory provides a new theoretic framework of evaluating the timeless of communications \cite{S.Kaul.M.2011}.
     Distinct from the traditional delay measure which only considers the latency to complete the transmission of packets regardless of the packet generation machanism, AoI is defined as the difference between the current epoch and the generation epoch of the latest received packet.
         That is, AoI considers the effects of both the information source and transmission channel.
In doing so, AoI characterizes the freshness of the available packet at the receiver more precisely.
    By modeling the arrivals and transmissions of packets as a queueing system, the AoI of the various systems can be obtained explicitly, such as the M/M/1 queue and the M/D/1 queue, with the first-come-first-service or the last-come-first-service policy, respectively \cite{S.Kaul.M.2011}.
In \cite{KAM.C.2018}, the author explored the impact of service rate on the average AoI of both systems with a fixed deadline and a random deadline.
    In addition to this kind of timeliness characterizations, we can also optimize the packet scheduling based on the AoI theory.
For example, the optimal link scheduling under some throughput and energy constraints was studied in  \cite{KADOTA.I.2019}, \cite{TANG.Haoyue.2019}.
    The   peak AoI and average AoI minimizing scheduling of multi-channel networks was investigated in \cite{SOMBABU.B.2020}.
In \cite{TALAK.R.2020}, \cite{TALAK.R.2017}, the authors constructed a feasible scheduling set by traversal and deduced an average-AoI minimizing scheduling strategy.

\subsection{Organizations}
This rest of the paper is organized as follows. In Section \ref{System Model}, we present the definitions of on-time reception, the channel model, and the source model.
        In Section \ref{Classical Probability Analysis}, we analyze the probability that each packet is received on time and also the on-time reception rate of the random transmission scheme.
In Section \ref{Alignment Transmission}, we present three packet controlling strategies. In Section \ref{Optimal Packet Scheduling}, we formulate an MDP optimization problem to solve the optimal packet scheduling policy.
    For a sequence of packet transmissions, we also derive the corresponding reward in theory in this section.
In Section \ref{Simulation results}, we present the simulation and numerical results on the on-time reception rates  over the fading channel, with both the random transmission scheme and the optimal packet scheduling policy.
    Finally, we conclude the paper in Section \ref{Conclusion}.

\section{System Model}\label{System Model}
\subsection{Definition of On Time}
Different from conventional in-time communications in which the packets are delivered with best-effort and are expected be received before a certain deadline, the on-time receptions (of the transmitted packets) studied in this paper require that each packet should be received exactly at its desired epoch, without any early arrivals or delays.
However, wireless channels are random and time-varying, and thus it is very difficult to guarantee that all of the transmitted packets could be received on-time. In this paper, we shall investigate how fading channels can support the on-time transmissions in terms of on-time reception ratio.

We consider the sequential transmissions of packets over a fading channel.
    We assume that the transmission of a packet starts from the beginning of a slot and is completed at the end of the slot. Due to the fading property of the channel, the \textit{transmission time} (i.e., the number of slots) to successfully deliver a packet is random. Suppose that the packets are intended to be received by the destination node at a sequence of preset slots (i.e., $\{{T}_{\text{tgt}}, 2{T}_{\text{tgt}}, 3{T}_{\text{tgt}}, \cdots \}$) with fixed intervals.
 The on-timeliness of the corresponding transmissions are defined as follows.

\begin{definition}\label{df:strict_ontime}
	The $m$-th packet is said to be received \textit{strictly on-time} if the packet is received by the destination node exactly in the $m{T}_{\text{tgt}}$-th slot.
\end{definition}

  As mentioned, strictly on-time transmission over fading channels is quite difficult so we allow the receptions of packets to deviate from the target slot with a maximum tolerance of $\delta$ slots.
 A slightly relaxed version of the on-timeliness is defined as follows.
\begin{definition}\label{df:delta_ontime}
	The $m$-th packet is said to be received $\delta$-\textit{on-time} if the packet is received by the destination node in any of the slots among $\{m{T}_{\text{tgt}}-\delta, m{T}_{\text{tgt}}-\delta+1, \cdots, m{T}_{\text{tgt}}+\delta\}$ (cf. Fig. \ref{fig:1_ontime_df}).
    Moreover, the period $\{m{T}_{\text{tgt}}-\delta, m{T}_{\text{tgt}}-\delta+1, \cdots, m{T}_{\text{tgt}}+\delta\}$ is referred to as the \textit{target reception range} of the $m$-th packet.
\end{definition}

It is clear that the $\delta$-on-time returns to the strictly on-time if we set the deviation tolerance to be $\delta=0$.

\captionsetup[figure]{labelformat={default},labelsep=period,name={Fig.}}
\begin{figure}[htp!]
  \centering
  \includegraphics[width=3.6in]{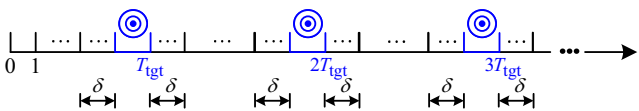}\\
  \caption{The on-time reception model. ${T}_{\text{tgt}}(\ge 1)$ is the preset interval between target reception epochs, $\delta \ge 0$ is the tolerance of deviations, $m{{T}_{\text{tgt}}}$ is the target reception slot of the $m$-th packet.}
  \label{fig:1_ontime_df}
\end{figure}

\subsection{Channel and Source Models}
We consider the packet transmissions over a fading channel with power gain distribution $f_\gamma(x)$.
    We denote the distance between the source and destination nodes as $d$, the path loss exponent as $\alpha$, and the transmit power of the source node as ${{P}_{\text{t}}}$.
In the $n$-th slot, the power of the received signal at the destination node can then be expressed as ${{P}_{\text{t},n}}={{{\gamma }_{n}}{{P}_{\text{t}}}}/{{{d}^{\alpha }}}$, in which ${\gamma }_{n}$ is the random power gain of the channel in the $n$-th slot.
    Thus, the signal-to-noise ratio (SNR) at the destination node can be expressed as
\begin{align}\label{eq:Rayleigh}
{{\rho }_{n}}=\frac{{{\gamma }_{n}}{{P}_{\text{t}}}}{{{d}^{\alpha }}\sigma^{2}},
\end{align}
in which $\sigma^{2}$ is the power of the Gaussian white noises.

We assume that the minimum SNR for the destination node to successfully decode the received packet is ${{V}_{\text{T}}}$.
    That is, the destination can successfully decode the packet from the received signal only if the corresponding SNR ${{\rho }_{n}}$ is larger than ${{V}_{\text{T}}}$.
Otherwise, the packet cannot be decoded and shall be retransmitted in the next slot.
    Thus, the probability that the destination node can decode a packet from the received signal can be expressed as
\begin{align}\label{rt:prb_suc}
    p =\Pr \left\{ {{\rho }_{n}}>{{V}_{\text{T}}} \right\}
    =\int_{\frac{{{V}_{\text{T}}}{{d}^{\alpha }}\sigma^{2}}{{{P}_{\text{t}}}}}^{+\infty }
            {{{f}_{\gamma }}\left( x \right)\text{d}x}.
\end{align}
It can be seen that the transmission time $S$ to successfully deliver a packet over the fading channel follows the geometric distribution with parameter $p$:
\begin{align}\label{eq:service time}
\Pr \left\{ S=j \right\}=p{{\left( 1-p \right)}^{j-1}},~~j=1,2,\ldots.
\end{align}

On the packet generations, we assume that the $(m+1)$-st packet will be generated immediately after the transmission completion of the $m$-th packet.
    After the generation of the packet, the source node begins to transmit the packet immediately.
    In particular, irrespective of the packet generation time, the desired reception time of the $(m+1)$-st packet is $(m+1)T_{\text{tgt}}$.

\subsection{On-Time Reception Rate}
The primary concern of this work is how many packets can be received on time, i.e., within their respective target reception ranges.
    In the transmission of a total number of $M$ packets, we {denote the number of packets received with} $\delta$-on-time as $\kappa_M$.
Then the \textit{on-time reception rate} ${{\varrho}_{M}}$ can be defined as
\begin{equation}\label{df:on_recpt_ratio}
    \varrho_M=\frac{\kappa_M}{M}.
\end{equation}

We would like to mention that the on-time reception rate is closely related to the length of the packet sequence $M$.
    Specifically, the larger $M$ is, the smaller the on-time reception rate would probably be.
This is because when more packets are transmitted, there would be more unexpectedly large transmission times, which makes the following packets more difficult to be received on-time.
    As will be shown in Section \ref{Simulation results}, the on-time reception rate of the random transmission scheme decreases moderately with $M$, even when $M$ is very large. In this paper, we will maximize the on-time reception rate of the system by scheduling the transmissions of packets.

\section{On-Time Reception Rate of Random Transmissions}\label{Classical Probability Analysis}
In this section, we consider the on-time performance of the transmissions over the fading channel in the absence of scheduling and controlling.
    By deriving the probability that each packet is received $\delta$-on time, we can also obtain the average number of packets received with $\delta$-on time and the corresponding on-time reception rate.

We denote the total number of packets to be transmitted as $M$, the transmission time of the $m$-th packet as ${{\tau}_{m}}$, and the probability that the $m$-th packet is received $\delta$-on time as $P\left( {{x}_{m}} \right)$ for $m=1,2,\ldots M$.

First, we consider the probability for the first ($m=1$) packet to be received $\delta$-on time and have
\begin{align}\label{eq_1}
P\left( {{x}_{1}} \right)=\Pr \left\{ {T}_{\text{tgt}}-\delta \le {{\tau}_{1}}\le {T}_{\text{tgt}}+\delta  \right\}.
\end{align}

Note that its transmission time satisfies ${{\tau}_{m}}\ge 1$  and follows the geometric distribution with parameter $p$ (see \eqref{eq:service time}).
    In case ${T}_{\text{tgt}}\le 1+\delta$, we have ${T}_{\text{tgt}}-\delta \le 1$ and \eqref{eq_1} is equivalent to
\begin{align}\label{eq:px1_1}
  P\left( {{x}_{1}} \right)& =\Pr \left\{ 1\leq {{\tau}_{1}}\le {T}_{\text{tgt}}+\delta  \right\} \nonumber \\
  & =\Pr \left\{ {{\tau}_{1}}=1 \right\}+\ldots +\Pr \left\{ {{\tau}_{1}}={T}_{\text{tgt}}+\delta  \right\} \nonumber \\
  & =p+\ldots +p{{\left( 1-p \right)}^{{T}_{\text{tgt}}+\delta -1}} \nonumber \\
  & =1-{{\left( 1-p \right)}^{{T}_{\text{tgt}}+\delta }}.
\end{align}

In case ${T}_{\text{tgt}}> 1+\delta$, we have
\begin{align}\label{eq:px1_2}
 P\left( {{x}_{1}} \right)& =\Pr \left\{ {T}_{\text{tgt}}-\delta \le {{\tau}_{1}}\le {T}_{\text{tgt}}+\delta  \right\} \nonumber \\
 & =\Pr \left\{ {{\tau}_{1}}\le {T}_{\text{tgt}}+\delta  \right\}-\Pr \left\{ {{\tau}_{1}}\le {T}_{\text{tgt}}-\delta -1 \right\} \nonumber \\
 & ={{\left( 1-p \right)}^{{T}_{\text{tgt}}-\delta -1}}-{{\left( 1-p \right)}^{{T}_{\text{tgt}}+\delta }}.
\end{align}

By combining \eqref{eq:px1_1} and \eqref{eq:px1_2},  the probability that the first packet is received $\delta$-on time can be expressed as
\begin{align}
P\left( {{x}_{1}} \right)=\left\{
 \begin{aligned}
  & 1-{{\left( 1-p \right)}^{{T}_{\text{tgt}}+\delta }}, &&{T}_{\text{tgt}}\le 1+\delta   \\
  & {{(1-p)}^{{T}_{\text{tgt}}-\delta -1}}-{{\left( 1-p \right)}^{{T}_{\text{tgt}}+\delta }}, &&{T}_{\text{tgt}}>1+\delta.   \\
\end{aligned} \right.
\end{align}

For the $m$-th packet, which is intended to be received within $\{m{T}_{\text{tgt}}-\delta, m{T}_{\text{tgt}}-\delta+1, \cdots, m{T}_{\text{tgt}}+\delta\}$, the total transmission time $\sum_{k=1}^m{\tau}_k$ follows the negative binomial distribution with parameter $p$.
    The following proposition describes the probability of a packet being received $\delta$-on time.

\begin{proposition}\label{prop:pxm}
For a sequence of $M$ packet transmissions over the fading channel, the probability of the $m$-th packet being received $\delta$-on time is given by
\begin{align}
& P\left( {{x}_{m}} \right)=\left\{ \begin{aligned}
  & \sum\limits_{k=m}^{m{{T}_{\text{tgt}}}+\delta }{C_{k-1}^{m-1}{{p}^{m}}{{\left( 1-p \right)}^{k-m}}},\quad m{{T}_{\text{tgt}}}\le m+\delta  \\
  & \sum\limits_{k=m{{T}_{\text{tgt}}}-\delta }^{m{{T}_{\text{tgt}}}+\delta }{C_{k-1}^{m-1}{{p}^{m}}{{\left( 1-p \right)}^{k-m}}}, m{{T}_{\text{tgt}}}>m+\delta,  \\
\end{aligned} \right.
\end{align}
in which $p$ is the probability of successful reception in a slot, ${T}_{\text{tgt}}$ is the target reception interval, $\delta$ is the deviation tolerance, and $C_n^k = {n\choose k} $ is the combination operator.
\end{proposition}

\begin{proof}
See Appendix \ref{proof of corollary 1}.
\end{proof}

\begin{figure}[htp!]
  \centering
  \includegraphics[width=3.7in]{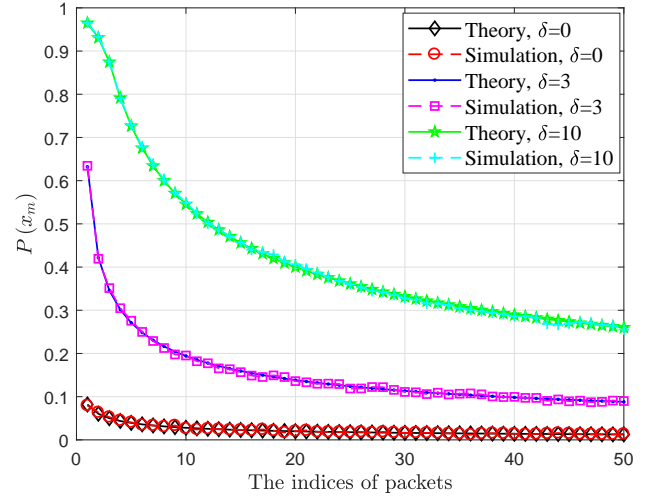}\\
  \caption{Probability being received $\delta$-on time. ($p=0.2$ and ${T}_{\text{tgt}}=5$).}
  \label{prop:pxmacket probability}
\end{figure}

From Proposition 1, we can observe the behavior of $P\left( {{x}_{m}} \right)$.
    In Fig. \ref{prop:pxmacket probability}, we compare $P\left( {{x}_{m}} \right)$  obtained by analytical and simulation results, in which $p=0.2$ and ${T}_{\text{tgt}}=5$.
    We observe that  $P\left( {{x}_{m}} \right)$ decreases with packet index $m$.
That is, the $P(x_m)$ of the $m$-th packet is no larger than that of previous packets.
    It is also seen that $P\left( {{x}_{m}} \right)$ increases with the deviation tolerance $\delta$.

We denote a subset of $k$ (not necessarily successive) packets out of $M$ as $x_{k}^{M} $. We denote the probability that $k$ out of the $M$ packets are received $\delta$-on time as $P\left( x_{k}^{M} \right)$.
    Among the $M$ packets, therefore, the statistical average number of packets received $\delta$-on time would be
\begin{align}\label{df:kappa}
    \kappa_M=\sum\limits_{k=1}^{M}{kP\left( x_{k}^{M} \right)}.
\end{align}

Moreover, as shown in the following theorem, $\kappa_M$ can be further expressed in terms of $P\left( x_{k} \right)$.

\begin{theorem}\label{classical}
For the transmissions of a sequence of $M$ packets over the fading channel, the probability that packets are received $\delta$-on time satisfies
\begin{align}\label{eq:expectation}
\sum\limits_{k=1}^{M}{kP\left( x_{k}^{M} \right)}=\sum\limits_{k=1}^{M}{P\left( {{x}_{k}} \right)},
\end{align}
 where $P\left( {{x}_{k}} \right)$ is the probability for the $k$-th packet to be received $\delta$-on time (cf. Proposition \ref{prop:pxm}) and $P\left( x_{k}^{M} \right)$ is the probability for $k$ out of the $M$ packets being received $\delta$-on time.
\end{theorem}

\begin{proof}
    See Appendix \ref{Proof of Theorem 1}.
\end{proof}

From Theorem \ref{classical}, we have
\begin{align}\label{eq:Classical probability total expected reward}
    \kappa_M=\sum\limits_{k=1}^{M}{kP\left( x_{k}^{M} \right)}=\sum\limits_{k=1}^{M}{P\left( {{x}_{k}} \right)}.
\end{align}
That is, the average number of packets that can be received $\delta$-on time is equal to the sum probability for each of the packets to be received $\delta$-on time.
    Since the probability ${P\left( {{x}_{k}} \right)}$ for each of the packets being received $\delta$-on-time has been given in Proposition \ref{prop:pxm}, we can calculate the on-time reception rate ${{\varrho }_{M}}={{{\kappa }_{M}}}/{M}$ by combing the results in Proposition \ref{prop:pxm}, Theorem \ref{classical}, and equation \eqref{eq:Classical probability total expected reward}.


\section{On-Time Reception Rate of Controlled Transmissions}\label{Alignment Transmission}
In this section, we present three controlling strategies to improve the on-time reception rate of the system.
    We first discuss the low on-time reception behavior of the random transmissions and then discuss a scheme to improve the on-time reception rate using delaying, dropping, and repeating the packets.

\subsection{Drawbacks of Random Transmissions}
In the random transmission scheme, any packet that is not received $\delta$-on time not only degrades system performance, but also  affects the transmissions of the subsequent packets.
    For example, if a packet is received before its target reception range, the probability that the next packet is received $\delta$-on time will also be reduced since there is more than enough time for its transmission so that the packet might be received earlier than the desired time.
If the transmission time of a packet is large so the packet is received after the target reception range, the probability for the subsequent packet being received $\delta$-on time may also be reduced (even to zero), since the remaining time for its transmission is shortened.
    Therefore, the on-time reception rate of the random transmissions is often relatively small.

\subsection{Controlling with Delaying, Dropping or Repeating}\label{subsec:ctl_policy}
We propose a scheme to control the transmission of  packets and improve the on-time reception rate of the system by delaying dropping, or repeating the transmission of packets.

\subsubsection{Delaying}
At the beginning of each packet transmission, the delaying strategy would delay the transmission for a period of ${{n}_{\text{d}}}\left( {{n}_{\text{d}}}=0,1,2,\ldots  \right)$ slots.
    This strategy is especially useful when the previous packet is received before its target reception range.
In particular, when the delayed time ${{n}_{\text{d}}}$ is set to $0$, the delaying strategy is equivalent to the random transmissions.

\subsubsection{Dropping}
Under the dropping strategy, the packets can be dropped on demand so that the next packet could be transmitted immediately.
    This strategy is very useful if the transmission time of the previous packet is so large that the subsequent packet completely misses the chance to be received $\delta$-on time.

\subsubsection{Repeating}
In case a packet is received before its target reception range, the repeat strategy allows the packet to be retransmitted. This is especially useful in the case the destination node is periodically awaken only for a short period.
We limit the retransmissions to a finite number of times and denote the \textit{maximum number of allowed retransmissions} as $n_\text{r}\left( n_\text{r}\ge0\right)$.
    In case $n_\text{r}=0$, the repeating strategy is equivalent to random transmissions.

In the following, we investigate the performance of the repeating strategy.
    Specifically, we consider a single packet ($M=1$) transmission over the fading channel.
 We denote the target reception interval as ${T}_{\text{tgt}}$, the deviation tolerance as $\delta$ and the maximum number of retransmissions as $n_\text{r}$.
    It should be noted that once the packet is received within its target reception range (i.e., $\delta$-on-time) or after the target reception range (i.e., have missed its chance), the transmission stops immediately and no more retransmission is needed.
 Under the repeating strategy, therefore, the transmission time ${S }_{n_\text{r}}$ of a packet can be expressed as
    \begin{equation}
        {S}_{n_\text{r}}=\sum\limits_{i=0}^{l}{{s}_{\text{r},i}},
    \end{equation}
    in which ${s}_{\text{r},i}$ is the transmission time of the $i$-th retransmission and follows the geometric distribution (see \eqref{eq:service time}), and $l$ is the random variable taking values among $\{0,1,2,\cdots, n_\text{r}\}$.
        It is clear that the packet could be received $\delta$-on-time if and only if $|{S }_{n_\text{r}}-{T}_{\text{tgt}}|\le\delta$.
  On the distribution of ${S }_{n_\text{r}}$, we further have the following result.


\begin{proposition}\label{Probability distribution function of repetition strategy}
Considering the transmission of a single packet ($M=1$) with the maximum number $n_\text{r}\ge 0 $ of retransmissions, the probability distribution function of the transmission time ${S }_{n_\text{r}}$ is given by
\begin{align}\label{corollary2.1}
   &\Pr \left\{{S }_{n_\text{r}}-{T}_{\text{tgt}}>j \right\}\nonumber \\
   &=\left\{ {
   \begin{aligned}
      &1,&& j<n_\text{r}-{T}_{\text{tgt}} \\
      &\sum\limits_{m=0}^{n_\text{r}}C_{{T}_{\text{tgt}}-\delta -1}^{m}{{p}^{m}}{{\left( 1-p \right)}^{j+{T}_{\text{tgt}}-m}},&& j\ge -1-\delta  \\
      &\sum\limits_{m=0}^{n_\text{r}}C_{j+{T}_{\text{tgt}}}^{m}{{p}^{m}}{{\left( 1-p \right)}^{j+{T}_{\text{tgt}}-m}},&& \text{else},  \\
  \end{aligned}} \right.
\end{align}
 if ${T}_{\text{tgt}}\ge 1+n_\text{r}+\delta$.
In the case  ${T}_{\text{tgt}} \le 1+\delta$, we have
\begin{align}\label{corollary2.2}
  \Pr \left\{ {S }_{n_\text{r}}-{T}_{\text{tgt}}>j \right\}& =\Pr \left\{ {S }_{0}-{T}_{\text{tgt}}>j \right\}\nonumber \\
  & =\left\{
  \begin{aligned}
  & {{\left( 1-p \right)}^{j+{T}_{\text{tgt}}}},&& j\ge -{T}_{\text{tgt}} \\
  & 1,&& j<-{T}_{\text{tgt}}. \\
\end{aligned} \right.
\end{align}
In case $1+\delta< {T}_{\text{tgt}}<1+n_\text{r}+\delta$, we have
\begin{align}\label{corollary2.3}
  \Pr \left\{ {S }_{n_\text{r}}-{T}_{\text{tgt}}>j \right\}& =\Pr \left\{ {{S }_{{{T}_{\text{tgt}}}-1-\delta }}-{T}_{\text{tgt}}>j \right\}\nonumber \\
  &=\left\{
  \begin{aligned}
  & {{\left( 1-p \right)}^{j+1+\delta }},&& j\ge -1-\delta  \\
  & 1,&& j<-1-\delta.  \\
  \end{aligned} \right.
\end{align}
\end{proposition}

\begin{proof}
    See Appendix \ref{proof of corollary 2}.
\end{proof}

\begin{figure}
  \centering
  \includegraphics[width=3.7in]{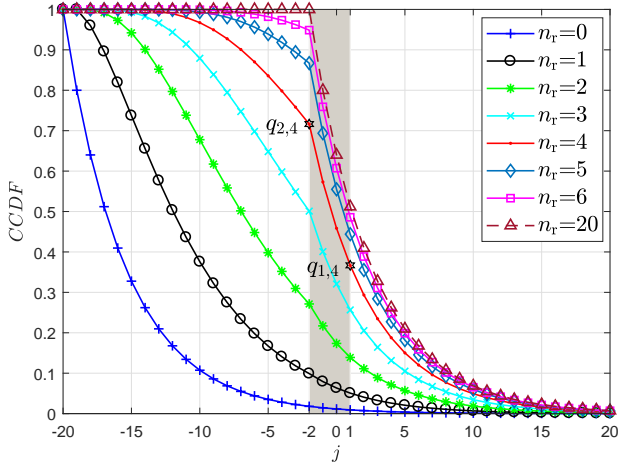}\\
  \caption{Complementary cumulative distribution functions (CCDF) of ${{S}_{n_\text{r}}}-{T}_{\text{tgt}}$  ($p=0.2$, ${{T}_{\text{tgt}}}=20$ and $\delta=1$).}
  \label{PDF}
\end{figure}

We see from Proposition \ref{Probability distribution function of repetition strategy} that the repeating strategy changes the distribution of the transmission time of each packet and improves the on-time performance of the transmissions.
    In fact, the probability that a packet is received $\delta$-on-time increases substantially by repeating the transmission of the packet for some times.
In Fig. \ref{PDF}, we present the complementary cumulative distribution functions (CCDF) of ${{S}_{n_\text{r}}}-{T}_{\text{tgt}}$ for various $n_\text{r}$ (the maximum numbers of retransmissions), in which the probability of successful transmission is set to $p=0.2$, the target reception interval is set to ${T}_{\text{tgt}}=20$, and the deviation tolerance is set to $\delta=1$.
    We plot the target reception range, i.e., the area in interval $(-2,1]$, by the shaded area.
We denote the y-coordinates of the intersections of each CCDF curve and the shaded area as $q_{1,n_\text{r}}$ and $q_{2,n_\text{r}}$.
    It is clear that $q_{2,n_\text{r}}-q_{1,n_\text{r}}=\Pr\{{{S}_{n_\text{r}}}-{T}_{\text{tgt}}>-2\}-\Pr\{{{S}_{n_\text{r}}}-{T}_{\text{tgt}}>1\}$ is the probability that the packet is received $\delta$-on-time, when a maximum $n_\text{r}$ retransmissions are allowed.
From Fig. \ref{PDF}, we observe that as $n_\text{r}$ increases, the probability for the packet to be received $\delta$-on-time will also increase, while the probabilities of the packets being received before (i.e., $1-q_{2,n_\text{r}}$) and after (i.e.,  $q_{1,n_\text{r}}$) the target reception range are, respectively, decreases and increases.
    However, it is clear that the uncertainty in the reception time is unavoidable due to the randomness of the fading channel.
For example, we observe that in case $n_\text{r}={T}_{\text{tgt}}=20$, the probability that a packet is received after the right boundary of the target reception range is still quite large.
    Thus, the gain in the probability of $\delta$-on-time by repeating the packets is also limited.

\section{Optimal Packet Scheduling}\label{Optimal Packet Scheduling}
In practical transmissions of a sequence of packets, a packet may either be received before or after its desired target reception range.
    To increase the probability of being received  $\delta$-on-time for the following packets, controlling strategies such as the delaying, dropping, and repeating should be considered.

In this section, we shall maximize the on-time reception rate of the system by modeling the optimal packet scheduling problem as an MDP problem.
    Using the MDP formulation, we can determine the optimal controlling strategy of a packet in an online manner.
That is, we shall determine the controlling strategy of a packet based on the  state of the system at the starting time of its transmission.
    In particular, we solve the optimal scheduling policy using the MDP-based iterative algorithm.

    To be specific, the state set, available actions, transition probabilities, reward functions, and optimal packet scheduling policy of the MDP problem are elaborated in the following subsections 1 to 5, respectively.

\setcounter{subsubsection}{0}
\subsubsection{States}
Since the controlling strategies are selected in the beginning of packet transmissions, we only need to consider the states of the system when a packet starts its transmission.

    To be specific, we define the \textit{state} $s_m$ of the system as the difference between the transmission starting time and the target reception time of the packet.
    For example, suppose the  $(m-1)$-st packet is received in slot $n-1$ and thus the $m$-th packet has a chance to be transmitted from slot $n$.
In slot $n$, the state of the system would then be $s_m=m{T}_{\text{tgt}}-n+1$,  in which $m{T}_{\text{tgt}}$ is the time when the packet is expected to be received.
    Since the previous packet $m-1$ may be received even later than the target reception time $m{T}_{\text{tgt}}$ of the $m$-th packet, the state $s_m$ could also be negative.
The state set, therefore, would be $\mathcal{S}=\mathbb{Z}$, i.e., the integer set.

Suppose the current state is $s_m=i$, the transmission time of the packet is $S$, and the next state $s_{m+1}=j$,  we further have
\begin{align}\label{eq:transfer}
j=i-S+{T}_{\text{tgt}},\quad \quad \quad i,j\in \mathcal{S}.
\end{align}

\subsubsection{Actions}
In the beginning of the transmission of a packet $m$, we schedule the packet by either delaying it by $n_{\text{d}}\ge0$ slots, dropping it, or repeating the transmission by $n_\text{r}\ge0$ times, which are referred to as taking an action $a$.
    The set of all possible actions is referred to as the action set $\mathcal{A}$.

\subsubsection{Transition probabilities}
For a given state $s_m=i$ and the corresponding action $a\in \mathcal{A}$, we denote the transition probability to state $s_{m+1}=j$ as ${{p}_{ij}}\left( a \right)=\Pr\{s_{m+1}=j|s_m=i\}$ and have
\begin{align}\label{eq:1}
\sum\limits_{j\in \mathcal{S}}{{{p}_{ij}}\left( a \right)}=1.
\end{align}

\begin{figure}[!t]
  \centering
  \includegraphics[width=3.6in]{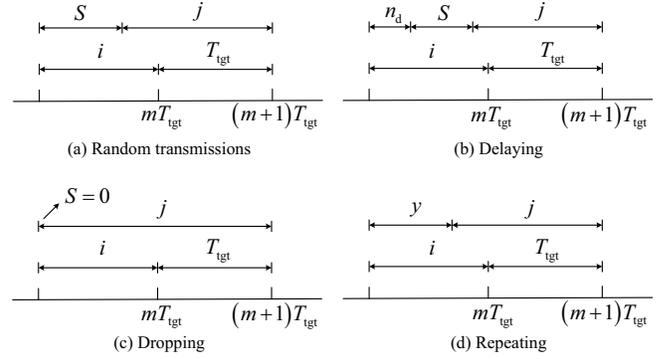}\\
  \caption{Illustration of controlling strategies, in which $y=i+{T}_{\text{tgt}}-j$.}
  \label{fig:Fig2_stratiges}
\end{figure}

By taking different actions, the transition probabilities are also different.
    First, we consider the transition probability of random transmissions.
By combing \eqref{eq:service time} and \eqref{eq:transfer}, and the fact $S\ge 1$, we have (cf. Fig. \ref{fig:Fig2_stratiges}(a))
\begin{align}\label{eq:Transition probability without control strategy}
{{p}_{ij}}=\left\{
\begin{aligned}
  & 0, &&   j\ge i+{T}_{\text{tgt}} \\
  & p{{\left( 1-p \right)}^{i-j+{T}_{\text{tgt}}-1}},&&j<i+{T}_{\text{tgt}}. \\
\end{aligned} \right.
\end{align}

Second, we consider the transitions of delaying the transmission by $n_{\text{d}}$ slots and denote the corresponding probability as ${{p}_{ij}}\left( {{n}_{\text{d}}} \right)$.
    Likewise, the transition probability of the delaying action is (see Fig. \ref{fig:Fig2_stratiges}(b))
\begin{align}\label{eq:Transition probability of delay strategy}
    {{p}_{ij}}\left( {{n}_{\text{d}}} \right)=\left\{
    \begin{aligned}
      & 0,&& j\ge i+{T}_{\text{tgt}}-{{n}_{\text{d}}} \\
      & p{{\left( 1-p \right)}^{i-{{n}_{\text{d}}}-j+{T}_{\text{tgt}}-1}},&& j<i+{T}_{\text{tgt}}-{{n}_{\text{d}}}. \\
    \end{aligned} \right.
\end{align}

Third, given that the current state is $s_m=i$ and the packet $m$ is dropped, its transmission time would be deterministic, i.e., $S=0$.
    In addition, the packet $m+1$ will be generated and then transmitted immediately.
Thus, the time for the transmission of the $(m+1)$-st packet to be received strictly on-time (i.e., its state) is $j=i+{T}_{\text{tgt}}$ (cf. Fig. \ref{fig:Fig2_stratiges}(c)).
    That is, the transition probability $p_{ij}$ from state $i$ to state $j$ is
\begin{align}\label{eq:Transfer probability of discard policy}
    {{p}_{ij}}=\left\{
    \begin{aligned}
      & 1,&&j=i+{T}_{\text{tgt}} \\
      & 0,&&j\ne i+{T}_{\text{tgt}}. \\
    \end{aligned} \right.
\end{align}

Fourth, when the packet is allowed to be repeated with a maximum number $n_\text{r}$ retransmissions, we denote the transition probability from state $i$ to state $j$ as $p_{ij}(n_\text{r})$, on which we have the following proposition.

\begin{proposition}\label{eq:Proof formula of reward function}
When a packet is allowed to be retransmitted for at most $n_\text{r}$ times and
if $i\le 1+\delta$, we have
\begin{align} \label{eq:Repeat strategy transition probability 1}
{{p}_{ij}}(n_\text{r})=\left\{
\begin{aligned}
  & 0,&& j\ge i+{T}_{\text{tgt}} \\
  & p{{(1-p)}^{i-j+{T}_{\text{tgt}}-1}},&&j<i+{T}_{\text{tgt}}; \\
\end{aligned} \right.
\end{align}
if $1+\delta <i< 1+\delta +n_\text{r}$, we have
\begin{align} \label{eq:Repeat strategy transition probability 2}
{{p}_{ij}}\left( n_\text{r} \right)=\left\{
 \begin{aligned}
  & 0,&&j\ge 1+\delta +{T}_{\text{tgt}} \\
  & p{{(1-p)}^{\delta -j+{T}_{\text{tgt}}}},&&j<1+\delta +{T}_{\text{tgt}}; \\
 \end{aligned} \right.
\end{align}
if $i\ge 1+\delta +n_\text{r}$, we have
\begin{align} \label{eq:Repeat strategy transition probability 3}
 &{{p}_{ij}}\left( n_\text{r} \right) \nonumber \\
 &=\left\{
 \begin{aligned}
  &0, &&y\le n_\text{r} \\
  &C_{y-1}^{n_\text{r}}{{p}^{1+n_\text{r}}}{{\left( 1-p \right)}^{y-1-n_\text{r}}}, &&n_\text{r}<y\le i-1-\delta   \\
  &\sum\limits_{m=0}^{n_\text{r}}{C_{i-1-\delta }^{m}{{p}^{1+m}}{{\left( 1-p \right)}^{^{y-1-m}}}}, &&y> i-1-\delta,
\end{aligned} \right.
\end{align}
in which $y=i-j+{T}_{\text{tgt}}$.
\end{proposition}

\begin{proof}
    See Appendix \ref{Proof of corollary 3}.
    As an intuitive explanation, $y=i-j+{T}_{\text{tgt}}$ is time reserved for the transmissions and the retransmissions of packet $m$  (see Fig. \ref{fig:Fig2_stratiges}(d)).
        As long as the packet is received at the destination before the target reception range $[m{T}_{\text{tgt}}-\delta,m{T}_{\text{tgt}}+\delta]$ and the number of retransmissions is less than $n_\text{r}$, the retransmission of packet $m$ continues.
\end{proof}
\subsubsection{Reward function}
When the state of the system transits from $s_m=i$ to $s_{m+1}=j$, we define the reward ${{r}_{ij}}$ of the system as
\begin{align}\label{eq:Definition of reward}
{{r}_{ij}}=\left\{
\begin{aligned}
  1,\quad& \text{if}\ {T}_{\text{tgt}}-\delta \le j\le {T}_{\text{tgt}}+\delta  \\
  0,\quad& \text{else}. \\
\end{aligned} \right.
\end{align}
That is, the total reward of the system (i.e. the number of packets received $\delta$-on time) will increase by one if the $m$-th packet is received $\delta$-on time.
    Otherwise, the reward is set to zero.

From a state $i$ and with an action $a$, we define the expected reward function as  $r\left( i,a \right)$, which can be calculated by
\begin{align}\label{eq:Definition of reward function}
        r\left( i,a \right)=\sum\limits_{j\in \mathcal{S}}{{{p}_{ij}}\left( a \right){{r}_{ij}}},
\end{align}
in which ${{p}_{ij}}\left( a \right)$ (see \eqref{eq:Transition probability without control strategy} to \eqref{eq:Repeat strategy transition probability 3}) is the state transition probability under action $a$.
    In the following, we shall investigate the reward functions of different controlling strategies case by case.

First, we denote the reward function of the random transmissions as $R\left( i \right)$.
    By combining equations \eqref{eq:Transition probability without control strategy}, \eqref{eq:Definition of reward}, and \eqref{eq:Definition of reward function}, we have
\begin{align}\label{eq:No control strategy reward function complex}
R\left( i \right)& =\sum\limits_{j={T}_{\text{tgt}}-\delta }^{{T}_{\text{tgt}}+\delta }{{{p}_{ij}}}\nonumber \\
& =\left\{
 \begin{aligned}
  & 0, &&i\le -\delta  \\
  & \sum\limits_{j={T}_{\text{tgt}}-\delta }^{{T}_{\text{tgt}}+\delta }{p{{\left( 1-p \right)}^{i-j+{T}_{\text{tgt}}-1}}} ,&&i> \delta \\
  & \sum\limits_{j={T}_{\text{tgt}}-\delta }^{i+{T}_{\text{tgt}}-1}{p{{\left( 1-p \right)}^{i-j+{T}_{\text{tgt}}-1}}},&&\text{else},
 \end{aligned}\right.
\end{align}
which is equivalent to
\begin{equation}\label{eq:No control strategy reward function}
R\left( i \right)=\left\{
\begin{aligned}
  & 0, && i\le -\delta  \\
  & {{\left( 1-p \right)}^{i-1-\delta }}\left[ 1-{{\left( 1-p \right)}^{1+2\delta }} \right], && i>\delta  \\
  & 1-{{\left( 1-p \right)}^{i+\delta }}, &&\text{else}.  \\
\end{aligned} \right.
\end{equation}

Second, we denote the reward function of the delaying  strategy as ${{R}_{\text{DL}}}\left( i,{{n}_{\text{d}}} \right)$, where ${n}_{\text{d}}$ is the delay time.
    By replacing $i$ in \eqref{eq:No control strategy reward function} with $i-{n}_{\text{d}}$, we have
\begin{align}\label{eq:Reward function of delay strategy}
& {{R}_{\text{DL}}}\left( i,{{n}_{\text{d}}} \right)\nonumber \\
& =\left\{
  \begin{aligned}
  & 0, &&i\le {{n}_{\text{d}}}-\delta  \\
  & {{\left( 1-p \right)}^{i-{{n}_{\text{d}}}-\delta -1}}\left[ 1-{{\left( 1-p \right)}^{1+2\delta }} \right],&&i>\delta +{{n}_{\text{d}}} \\
  & 1-{{\left( 1-p \right)}^{\delta +i-{{n}_{\text{d}}}}},&&\text{else}. \\
  \end{aligned}\right.
\end{align}

Third, we denote the reward function of the dropping strategy as  ${{R}_{\text{DP}}}\left( i \right)$.
    Since a dropped packet can never be received on time, we immediately have
\begin{align}\label{eq:Reward function of drop policy}
{{R}_{\text{DP}}}\left( i \right)=0.
\end{align}

Fourth, we denote the reward function of the repeating strategy as ${{R}_{\text{RP}}}\left( i, n_\text{r} \right)$, in which $n_\text{r}$ is the maximum allowed retransmissions.
    From \eqref{eq:Definition of reward} and \eqref{eq:Definition of reward function}, we have
\begin{align}\label{eq:rwd:dp_1}
{{R}_{\text{RP}}}\left( i,n_\text{r} \right)=\sum\limits_{j={T}_{\text{tgt}}-\delta }^{{T}_{\text{tgt}}+\delta }{{p}_{ij}}\left( n_\text{r} \right),
\end{align}

In case $i\le 1+\delta$, from \eqref{eq:Repeat strategy transition probability 1}, we have
\begin{align}
{{R}_{\text{RP}}}\left( i,n_\text{r} \right)
& =\left\{
  \begin{aligned}
  & 0,&&i\le-\delta \\
  & \sum\limits_{j={T}_{\text{tgt}}-\delta }^{i+{T}_{\text{tgt}}-1}{p{{\left( 1-p \right)}^{i-j+{T}_{\text{tgt}}-1}}},&&i>-\delta. \\
  \end{aligned} \right.
\end{align}

In case $1+\delta <i\le 1+\delta +n_\text{r}$, from \eqref{eq:Repeat strategy transition probability 2}, we have
\begin{align}
{{R}_{\text{RP}}}\left( i,n_\text{r} \right)
    =\sum\limits_{j={T}_{\text{tgt}}-\delta }^{{T}_{\text{tgt}}+\delta }{p{{\left( 1-p \right)}^{\delta -j+{T}_{\text{tgt}}}}}.
\end{align}

In case  $i>1+\delta +n_\text{r}$, from \eqref{eq:Repeat strategy transition probability 3}, we have
\begin{align}\label{eq:rwd:dp_4}
{{R}_{\text{RP}}}\left( i,n_\text{r} \right)
    & =\sum\limits_{j={T}_{\text{tgt}}-\delta }^{{T}_{\text{tgt}}+\delta }{\sum\limits_{m=0}^{n_\text{r}}{C_{i-1-\delta }^{m}{{p}^{1+m}}{{\left( 1-p \right)}^{i-j+{T}_{\text{tgt}}-1-m}}}}.
\end{align}

By combing \eqref{eq:rwd:dp_1}$\sim$\eqref{eq:rwd:dp_4}, the reward function of the repeating strategy can be expressed as
\begin{align}\label{eq:Reward function of repeated policy}
& {{R}_{\text{RP}}}\left( i,n_\text{r} \right)=\left\{
  \begin{aligned}
  & 0,\ \ \ \ \ \ \ \ \ \ \ \ \ \ \ \ \quad \ \ \ \ \ i\le -\delta  \\
  & 1-{{\left( 1-p \right)}^{i+\delta }},\qquad\ \ -\delta <i\le 1+\delta  \\
  & 1-{{\left( 1-p \right)}^{1+2\delta }},\ \ \ \ \ \ 1+\delta <i\le 1+\delta +n_\text{r} \\
  & D\left[ 1-{{\left( 1-p \right)}^{1+2\delta }} \right],  i>1+\delta +n_\text{r}, \\
  \end{aligned} \right.
\end{align}
in which $D=\sum\limits_{m=0}^{n_\text{r}}{C_{i-1-\delta }^{m}{{p}^{m}}{{\left( 1-p \right)}^{i-1-\delta -m}}}$.

\subsubsection{Optimal Packet Scheduling Policy}
A packet scheduling policy $\boldsymbol{\pi}$ is a rule for choosing actions (i.e., delaying, dropping, or repeating) for each packet, i.e., a mapping from the state space $\mathcal{S}$ to the action space $\mathcal{A}$.
    Specifically, for each state $s=i$ (can be negative), the corresponding element of $\boldsymbol{\pi}$ specifies which action should be taken for the current packet.
In case the packet should be delayed, ${\pi}_i$ also indicates how long it should be delayed, i.e., determining $n_{\text{d}}$; in case the packet should be repeated, ${\pi}_i$ also indicates how many times it could be retransmitted, i.e., determining $n_\text{r}$; in case the packet should be randomly transmitted or dropped, no other parameters are needed to be determined.

For a sequence of packet transmissions, we  seek such a policy $\boldsymbol{{{\pi }^{*}}}$ that maximizes the average reward of the system with any initial state ${s}_1=i$.
    That is,
\begin{align}\label{eq:Average reward of the system}
{\boldsymbol{\pi}}^{*} =\arg \underset{\boldsymbol{\pi}}{\mathop{\max }}\,\mathbb{E}\left[\frac{1}{M} \sum\limits_{m=1}^{M}{R\left( {{{s}}_{m}},{{a}_{m}} \right)}\left| {{{s}}_{1}}=i \right.\right],
\end{align}
in which ${{{s}}_{m}}\left( m=1,2,\ldots ,M,{{{s}}_{m}}\in \mathcal{S} \right)$ is the state of the $m$-th packet, ${{a}_{m}}\left( m=1,2,\ldots ,M,{{a}_{m}}\in \mathcal{A} \right)$ is the action assigned for the $m$-th packet, $R\left( {{{s}}_{m}},{{a}_{m}} \right)$ is the reward function of the $m$-th packet when the state is ${{{s}}_{m}}$ and the action is ${{a}_{m}}$.
    From \eqref{eq:Definition of reward function}, we know that the reward function $R\left( {{{s}}_{m}},{{a}_{m}} \right)$ is also the probability that the packet is received $\delta$-on time after taking action ${{a}_{m}}$ in state ${{{s}}_{m}}$ so that we define the corresponding cost as
\begin{align}\label{eq:Cost function definition}
C\left( {{{s}}_{m}},{{a}_{m}} \right)=1-R\left( {{{s}}_{m}},{{a}_{m}} \right),
\end{align}
which is non-negative.
Therefore, we can also find out the optimal packet scheduling policy $\boldsymbol{{{\pi }^{*}}}$ by minimizing the following average cost.
\begin{align}\label{eq:cost}
{{V}_{{\boldsymbol{\pi}^{*}}}}\left( {i} \right)=\underset{\boldsymbol{\pi}}{\mathop{\min }}\,\mathbb{E}\left[\frac{1}{M} \sum\limits_{m=1}^{M}{C\left( {{{s}}_{m}},{{a}_{m}} \right)}\left| {{{s}}_{1}}=i \right. \right]
\end{align}

As shown in \cite[Chap. 6.7, \emph{Theorem} 6.17]{S.M.Ross1970}, \eqref{eq:cost} can be solved by the following functional equation,
\begin{align}\label{eq:cost conversion}
g+h\left( {i} \right)=\underset{a\in \mathcal{A}}{\mathop{\min }}\,\left\{ C\left( {s},a \right)+\sum\limits_{{j}\in \mathcal{S}}{{{p}_{{i}{j}}}\left( a \right)h\left( {j} \right)} \right\},
\end{align}
in which $g$ is a constant, $h\left( {i} \right)$ is a bounded function, ${{p}_{{i}{j}}}\left( a \right)$ is the state transition probability of the packet from state ${i}$ to state ${j}$ when  action $a$ is taken.

 It is noted, however, that \eqref{eq:cost conversion} is not a contraction mapping \cite[Chap. 6.4, \emph{Theorem} 6.10]{S.M.Ross1970}.
    Thus, the searching process with \eqref{eq:cost conversion}  may not converge or converge very slowly.
 This motivates us to consider an alternative expected total $\alpha$-discounted cost as shown below.
\begin{align}\label{eq:alpha-cost}
{{V}_{\boldsymbol{\pi} _{\alpha }^{*}}}\left( {i} \right)=\underset{\boldsymbol{\pi} }{\mathop{\min }}\,\mathbb{E}\left[\frac{1}{M} \sum\limits_{m=1}^{M}{{{\alpha }^{m}}C\left( {{{s}}_{m}},{{a}_{m}} \right)\left| {{{s}}_{1}}=i \right.} \right]
\end{align}
for all $i\in \mathcal{S}$, in which $0<\alpha <1$ is a discounting factor.
    Moreover, the $\alpha$-optimal policy $\boldsymbol{\pi} _{\alpha }^{*}$ and the $\alpha$-optimal cost function ${{V}_{\alpha }}\left( {i} \right)$ satisfies \cite[Chap. 6.7, (24)]{S.M.Ross1970},
\begin{align}\label{eq:Iteration cost}
{{V}_{\alpha }}\left({i} \right)=\underset{a\in \mathcal{A}}{\mathop{\min }}\,\left\{ C\left( i,a \right)+\alpha \sum\limits_{{j}\in \mathcal{S}}{{{p}_{{i}{j}}}\left( a \right){{V}_{\alpha }}\left( {j} \right)} \right\}.
\end{align}
Particularly, the following theorem shows that as $\alpha$ approaches unity, $\boldsymbol{\pi} _{\alpha }^{*}$ would converge to ${{\boldsymbol{\pi} }^{*}}$.

\begin{theorem}\label{Optimal policy}
For some sequence ${{\alpha }_{n}}\to 1$, we have $h\left( {s} \right)=\underset{n\to \infty }{\mathop{\lim }}\,{{V}_{{{\alpha }_{n}}}}\left( {s} \right)-{{V}_{{{\alpha }_{n}}}}\left( {{{s}}_{1}} \right)$, $g=\underset{\alpha \to 1}{\mathop{\lim }}\,\left( 1-\alpha  \right){{V}_{\alpha }}\left( {{{s}}_{1}} \right)$, for any fixed reference state ${{{s}}_{1}}$. In particular, \eqref{eq:cost} and \eqref{eq:alpha-cost} share the same optimal policy.
\end{theorem}

\begin{proof}
Based on the state transition probabilities given in \eqref{eq:Transition probability of delay strategy}, \eqref{eq:Transfer probability of discard policy}, \eqref{eq:Repeat strategy transition probability 1}, \eqref{eq:Repeat strategy transition probability 2} and \eqref{eq:Repeat strategy transition probability 3}, it can be seen that each state can reach all other states directly or indirectly through some intermediate states, which means that the Markov chain is irreducible. According to \cite[Chap. 6.8, \emph{Corollary} 6.20]{S.M.Ross1970}, ${{V}_{\alpha }}\left( {s} \right)-{{V}_{\alpha }}\left( {{{s}}_{1}} \right)$ would be uniformly bounded, and hence the conditions of \cite[Chap. 6.7, \emph{Theorem} 6.17]{S.M.Ross1970} are satisfied, which yield the results in \emph{Theorem} \ref{Optimal policy} immediately.
\end{proof}

Theorem \ref{Optimal policy} shows that \eqref{eq:Average reward of the system}, \eqref{eq:cost} and \eqref{eq:alpha-cost} have the same optimal scheduling policy.
    Thus, the optimal scheduling policy of the system can be calculated by the matrix iteration method.
To be specific, for each state $i$, we shall calculate the expected costs $C\left( i,a \right)+\alpha \sum_{{j}\in \mathcal{S}} {{{p}_{{i}{j}}}\left( a \right){{V}_{\alpha }}( {j}) }$ for each action $a\in \mathcal{A}$, including the random transmission, dropping the packet, delaying the packet for some slots ($n_{\text{d}}=1,2,\cdots, n_{\text{d}}^{\max }$), or retransmit the packet for some times ($n_\text{r}=1,2,\cdots,n_{\text{r}}^{\max }$).
    With the obtained expected costs, we can determine the best action and update the cost ${{V}_{\alpha }}\left({i} \right)$ with \eqref{eq:Iteration cost}.
    In particular, it was shown in \cite[Chap. 6.8]{S.M.Ross1970} that the mapping shown in \eqref{eq:Iteration cost} is contract mapping.
We denote the vector of all the states as $\boldsymbol{s}$ and the cost vector as ${\boldsymbol{V}_{\alpha }}$.
    By applying \eqref{eq:Iteration cost} to $\boldsymbol{s}$ (which is done element by element) iteratively, the cost vector ${\boldsymbol{V}_{\alpha }}$ would then converge to the optimal cost vector while the obtained actions are all optimal for the corresponding states, as shown in {Algorithm} \ref{Alg: Optimal scheduling strategy algorithm}.

Since the state of a packet is the difference between its transmission starting time and its target reception time, the state space $\mathcal{S}$ is often infinitely large.
     The probability for the state of packets to be very large or small, however, is very small and can be neglected.
 Thus, we shall limit the state space to the set of integers within the finite range $\left[ {\iota }_{\min },  {\iota }_{\max } \right]$, so that we can solve the problem more efficiently.
    In this case, the number of desirable states is ${\iota }_{\max }-{\iota }_{\min }+1$.

\begin{algorithm}[!t]
\algsetup{linenosize=\small}
\scriptsize
\caption{Solving the optimal packet scheduling policy}
\begin{algorithmic}[1]\label{Alg: Optimal scheduling strategy algorithm}
\STATE \textbf{Input}: \\
  cost matrix ${{\textbf{C}}_{\text{DL}}}$ and transition probability matrix ${{\textbf{P}}_{\text{DL}}}$ of the delaying strategy;\\
  cost vector ${{\boldsymbol{C}}_{\text{DP}}}$ and transition probability matrix ${{\textbf{P}}_{\text{DP}}}$ of the dropping strategy;\\
 cost matrix ${{\textbf{C}}_{\text{RP}}}$ and transition probability matrix ${{\textbf{P}}_{\text{RP}}}$ of the repeating strategy;\\
\STATE \textbf{Initialization}:\\
 Set iteration error to $\Delta v=+\infty$, $\varepsilon ={{10}^{-3}}$; \\
 Initialize the cost function vector ${{\boldsymbol{V}}_{\alpha }}=\text{zeros}\left( {\iota }_{\max }-{\iota }_{\min }+1,1 \right)$;\\
 Initialize the policy vector ${\boldsymbol{\pi}_{\alpha }^*}=\text{zeros}\left( {\iota }_{\max }-{\iota }_{\min }+1,1 \right)$;
\STATE \textbf{Iteration}:\\
\quad \textbf{while}: $\Delta v>\varepsilon$, \textbf{do}\\
\quad\quad ${{\boldsymbol{f}}_{\text{DP}}}={{\boldsymbol{C}}_{\text{DP}}}+\alpha {{\textbf{P}}_{\text{DP}}}{{\boldsymbol{V}}_{\alpha }}$;\\
\quad\quad \textbf{for} ${{n}_{\text{d}}}=0$ to $n_{\text{d}}^{\max }$ \textbf{do}\\
\quad\quad \quad ${{\textbf{F}}_{\text{DL}}}\left( :,{{n}_{\text{d}}}+1 \right)={{\textbf{C}}_{\text{DL}}}\left( :,{{n}_{\text{d}}}+1 \right)+\alpha {{\textbf{P}}_{\text{DL}}}\left( :,:,{{n}_{\text{d}}}+1 \right){{\boldsymbol{V}}_{\alpha }}$;\\
\quad\quad \textbf{end} \textbf{for}\\
\quad\quad \textbf{for} $n_\text{r}=0$ to $n_{\text{r}}^{\max }$ \textbf{do}\\
\quad\quad\quad ${{\textbf{F}}_{\text{RP}}}\left( :,n_\text{r}+1 \right)={{\textbf{C}}_{\text{RP}}}\left( :,n_\text{r}+1 \right)+\alpha {{\textbf{P}}_{\text{RP}}}\left( :,:,n_\text{r}+1 \right){{\boldsymbol{V}}_{\alpha }}$;\\
\quad\quad \textbf{end} \textbf{for}\\
\quad\quad ${{\boldsymbol{V}}^{\text{old}}}={{\boldsymbol{V}}_{\alpha }}$;\\
\quad\quad $\textbf{S}=\left[ {{\boldsymbol{f}}_{\text{DP}}},{{\textbf{F}}_{\text{DL}}},{{\textbf{F}}_{\text{RP}}} \right]$;\\
\quad\quad $\left[ {{\boldsymbol{V}}_{\alpha }},{\boldsymbol{\pi}_{\alpha }^*} \right]=\min \left( \textbf{S},2 \right)$;  \%find the minimum over the 2-nd dimenssion\\
\quad\quad $\Delta v=\max (|{{\boldsymbol{V}}_{\alpha }}-{{\boldsymbol{V}}^{\text{old}}}|)$;\\
\quad\textbf{end} \textbf{while}
\STATE \textbf{Output}: ${{\boldsymbol{V}}_{\alpha }}$, ${\boldsymbol{\pi}_{\alpha }^*}$.
\end{algorithmic}
\end{algorithm}

 From \eqref{eq:Transition probability of delay strategy}, \eqref{eq:Transfer probability of discard policy}, \eqref{eq:Repeat strategy transition probability 1}, \eqref{eq:Repeat strategy transition probability 2}, and \eqref{eq:Repeat strategy transition probability 3} we can explicitly express the transition matrices of delaying, dropping and repeating, which are denoted, respectively, as ${{\textbf{P}}_{\text{DL}}}$, ${{\textbf{P}}_{\text{DP}}}$, and ${{\textbf{P}}_{\text{RP}}}$.
    In particular, ${{\textbf{P}}_{\text{DL}}}$ and ${{\textbf{P}}_{\text{RP}}}$ are three-dimensional matrices.
In ${{\textbf{P}}_{\text{DL}}}$, the first and the second dimensions represent the states before and after the transition, while the third dimension represents the number of slots the packets are delayed, i.e., $n_{\text{d}}$.
    Likewise, the third dimension of ${{\textbf{P}}_{\text{RP}}}$  represents the maximum allowed number of retransmission, i.e., $n_\text{r}$.

From \eqref{eq:Reward function of delay strategy}, \eqref{eq:Reward function of drop policy}, \eqref{eq:Reward function of repeated policy}, and \eqref{eq:Cost function definition}, we can also obtain the cost functions of the three strategies, i.e., ${{\textbf{C}}_{\text{DL}}}$, ${{\boldsymbol{C}}_{\text{DP}}}$ and ${{\textbf{C}}_{\text{RP}}}$.
     Matrices ${{\textbf{C}}_{\text{DL}}}$ and ${{\textbf{C}}_{\text{RP}}}$ are two-dimensional matrices defining the costs for each state and each $n_\text{d}$ and $n_\text{r}$.

  Finally, $({\iota }_{\max }-{\iota }_{\min }+1)\times1$ vector of optimal packet scheduling policy  ${\boldsymbol{\pi}_{\alpha }^*}$ can be obtained by Algorithm \ref{Alg: Optimal scheduling strategy algorithm}.
        As shown in Theorem \ref{Optimal policy}, we have ${\boldsymbol{\pi}_{\alpha }^*}={\boldsymbol{\pi}^*}$, which specifies the actions for all the states.

With the obtained optimal scheduling policy ${\boldsymbol{\pi}_{\alpha }^*}$, which can be expressed by a state-action mapping table, we can then find out the optimal action (i.e., transmit it without control, delay it, drop it, or repeat it) of each packet based on its current state.
    As shown in our simulations in Section \ref{Simulation results}, the corresponding $\delta$-on-time reception rate achieves the optimal reward of the system exactly.

\subsection{Theoretical Analysis of  Expected Total Rewards}
In this section, we will analyze the total expected reward of two systems, in which the random transmission strategy and the packet scheduling is used respectively.
    In particular, the system with the packet scheduling would optimize the controlling strategy of each packet by maximizing the expected total reward of the system.

\subsubsection{System with Random Transmissions}
For the system using the random transmission strategy, the expected reward of a single transition from state $i$ can be expressed as
\begin{align}\label{eq:r_i}
R\left( i \right)=\sum\limits_{j\in \mathcal{S}}{{{p}_{ij}}{{r}_{ij}}},\ \ \ i\in \mathcal{S},
\end{align}
in which ${{p}_{ij}}$ (cf. \eqref{eq:Transition probability without control strategy}) and ${{r}_{ij}}$ (cf. \eqref{eq:Definition of reward}) are, respectively, the probability and the reward of the transition from $i$ to $j$.
    We denote the vector of expected transition rewards of all the states as $\boldsymbol{R}={{\left[ R\left( {\iota }_{\min } \right),R\left( {\iota }_{\min }+1 \right),\ldots ,R\left( {\iota }_{\max } \right) \right]}^{\mathbf{T}}}$.

We denote the expected total reward of a sequence of $m$ transitions from state $i$ as $v_{m}^{'}\left( i \right)$, which can be calculated based on $v_{m-1}^{'}\left( j \right)$ and ${r}_{ij}$ as
\begin{align}
v_{m}^{'}\left( i \right)& =\sum\limits_{j\in \mathcal{S}}{{{p}_{ij}}\left[ {r}_{ij}+v_{m-1}^{'}\left( j \right) \right]}\nonumber \\
& =R\left( i \right)+\sum\limits_{j\in \mathcal{S}}{{{p}_{ij}}v_{m-1}^{'}\left( j \right)},i\in \mathcal{S}.
\end{align}
  We denote the vector of the expected $m$-transition rewards of all the states as $\boldsymbol{V}_{m}^{'}={{\left[ v_{m}^{'}\left( {\iota }_{\min } \right),v_{m}^{'}\left( {\iota }_{\min }+1 \right),\ldots ,v_{m}^{'}\left( {\iota }_{\max } \right) \right]}^{\mathbf{T}}}$.
It is clear that
\begin{align}\label{eq:v1}
\boldsymbol{V}_{1}^{'}& =\boldsymbol{R}, \\
\label{eq:vm}
\boldsymbol{V}_{m}^{'}& =\boldsymbol{R}+\textbf{P}\boldsymbol{V}_{m-1}^{'},m=2,3,\ldots ,M,
\end{align}
in which $\textbf{P}$ is the state transition probability matrix of the random transmission strategy (cf. \eqref{eq:Transition probability without control strategy}).
    Starting from \eqref{eq:v1}, we repeatedly use \eqref{eq:vm} to obtain the expected total reward vector $\boldsymbol{V}_{M}^{'}$ of a sequence of $M-1$ state transitions  (Algorithm \ref{Total expected reward algorithm}).
Moreover, $\boldsymbol{V}_{M}^{'}$ is also the expected total reward of the system for transmitting $M$ packets to the destination node.
    \begin{remark}
        Note that the transmission of the first packet starts from the first slot and the corresponding initial state is $s_1 = {{T}_{\text{tgt}}} $.
            When the transmission of all of $M$ packets have been completed, the expected total reward of the system would then be $v_{M}^{'}({{T}_{\text{tgt}}})$, which can be obtained by Algorithm \ref{Total expected reward algorithm}.
        Note also that the expected total reward of a system using random transmissions equals to the  number $\kappa_M$ (cf. \eqref{df:kappa}) of packets received $\delta$-on-time, which can be obtained through classical probability methods, as shown equations \eqref{df:kappa} and \eqref{eq:expectation}, in Section \ref{Classical Probability Analysis}.
    In particular, it can be verified through simulations that the $v_{M}^{'}({{T}_{\text{tgt}}})$ obtained by Algorithm \ref{Total expected reward algorithm} equals to $\kappa_M$ exactly.
    \end{remark}


\subsubsection{System with Scheduling}
Likewise, we calculate the expected total reward of the system with packet scheduling iteratively, as shown in Algorithm \ref{Total expected reward algorithm}.

We note that the transition probability matrices ${{\textbf{P}}_{\text{DL}}}(n_\text{d})=[{{p}_{ij}} ({{n}_{\text{d}}}) ]$, ${{\textbf{P}}_{\text{DP}}}=[p_{ij}]$,
${{\textbf{P}}_{\text{RP}}}(n_\text{r})=[{{p}_{ij}} (n_\text{r}) ]$
of the delaying strategy, the dropping strategy, and the repeating strategy are given, respectively, by \eqref{eq:Transition probability of delay strategy}, \eqref{eq:Transfer probability of discard policy} and \eqref{eq:Repeat strategy transition probability 1} to \eqref{eq:Repeat strategy transition probability 3}.
    For each state $i\in \mathcal{S}$ and each chosen strategy, therefore, the expected reward of the next transition $R\left( i \right)$ can be calculated by \eqref{eq:r_i}, in which $p_{ij}$ is replaced by ${{p}_{ij}} ({{n}_{\text{d}}})$, $p_{ij}$, and ${{p}_{ij}} (n_\text{r}) $, respectively.

We denote the  vector of expected $m$-transition rewards of the system with scheduling as $\boldsymbol{V}_{m}={{\left[ v_{m}\left( {\iota }_{\min } \right),v_{m}\left( {\iota }_{\min }+1 \right),\ldots ,v_{m}\left( {\iota }_{\max } \right) \right]}^{\mathbf{T}}}$.
Given the expected $m$-transition reward vector $\boldsymbol{V}_{m}$, we shall first estimate the expected $(m+1)$-transition rewards for all the cases when the delaying (for all ${{n}_{\text{d}}}$), dropping, and the repeating (for all $n_\text{r}$) strategy are used.
Specifically, we have
\begin{align}
{{\textbf{F}}_{\text{DL}}}\left( :,{{n}_{\text{d}}}+1 \right)=&{{\textbf{R}}_{\text{DL}}}\left( :,{{n}_{\text{d}}}+1 \right)+ {{\textbf{P}}_{\text{DL}}}\left( :,:,{{n}_{\text{d}}}+1 \right){{\boldsymbol{V}}_{m }} \\
{{\boldsymbol{f}}_{\text{DP}}}=&{{\boldsymbol{R}}_{\text{DP}}}+ {{\textbf{P}}_{\text{DP}}}{{\boldsymbol{V}}_{m}} \\
{{\textbf{F}}_{\text{RP}}}\left( :,n_\text{r}+1 \right)=&{{\textbf{R}}_{\text{RP}}}\left( :,n_\text{r}+1 \right)+ {{\textbf{P}}_{\text{RP}}}\left( :,:,n_\text{r}+1 \right){{\boldsymbol{V}}_{m }},
\end{align}
     for ${{n}_{\text{d}}}=0,1,\cdots,{{n}_{\text{d}}^{\max}}$ and $n_\text{r}=0,1,\cdots,n_{\text{r}}^{\max }$.
For each state $i$, therefore, we have obtained the expected total reward for all controlling strategies (i.e., delaying, dropping, and repeating) and parameters (i.e., ${{n}_{\text{d}}}$ and ${n}_{\text{r}}$).
    By searching the maximum reward among $\{{{\textbf{F}}_{\text{DL}}}( i,1 ),\cdots,{{\textbf{F}}_{\text{DL}}}( i,{{n}_{\text{d}}^{\max}}+1 ), {{\boldsymbol{f}}_{\text{DP}}}(i), {{\textbf{F}}_{\text{RP}}}( i,1 ),\cdots, {{\textbf{F}}_{\text{RP}}}( i, n_{\text{r}}^{\max }+1 )\}$, we can then determine the optimal controlling action and parameter.
With the obtained controlling strategy and parameter, we can further update the expected $(m+1)$-transition rewards $\boldsymbol{V}_{m+1}$ of system.
    As shown in Algorithm \ref{Total expected reward algorithm}, this process continues until the controlling strategies of all the packets has been determined and the expected total reward of the system with scheduling is  ${v}_{M}\left( {T}_{\text{tgt}} \right)$.

\begin{algorithm}[!t]
\algsetup{linenosize=\small}
\scriptsize
\caption{Total expected reward}
\begin{algorithmic}[1]\label{Total expected reward algorithm}
\STATE \textbf{Input}: \\
 reward vector $\boldsymbol{R}$ and transition probability matrix $\textbf{P}$ of random transmission;\\
 reward vector ${{\boldsymbol{R}}_{\text{DP}}}$ and transition probability matrix ${{\textbf{P}}_{\text{DP}}}$ of drop strategy;\\
 reward matrix ${{\textbf{R}}_{\text{DL}}}$ and transition probability matrix ${{\textbf{P}}_{\text{DL}}}$ of delay strategy;\\
 reward matrix ${{\textbf{R}}_{\text{RP}}}$ and transition probability matrix ${{\textbf{P}}_{\text{RP}}}$ of repeat strategy;\\
\STATE \textbf{Initialization}:\\
 Initialize optimal scheduling policy vector ${{\boldsymbol{V}}_{0}}=\text{zeros}\left( {\iota }_{\max }-{\iota }_{\min }+1,1 \right)$;\\
 Initialize random transmission vector ${\boldsymbol{V}_{0}^{'}}=\text{zeros}\left( {\iota }_{\max }-{\iota }_{\min }+1,1 \right)$;\\
\STATE \textbf{Iteration}:\\
\quad \textbf{for}: $m=1$ to $M$ \textbf{do}\\
\quad\quad $\boldsymbol{V}_{m}^{'}=\boldsymbol{R}+\textbf{P}\boldsymbol{V}_{m-1}^{'}$;\\
\quad\quad ${{\boldsymbol{f}}_{\text{DP}}}={{\boldsymbol{R}}_{\text{DP}}}+ {{\textbf{P}}_{\text{DP}}}{{\boldsymbol{V}}_{m-1}}$;\\
\quad\quad \textbf{for} ${{n}_{\text{d}}}=0$ to $n_{\text{d}}^{\max }$ \textbf{do}\\
\quad\quad \quad ${{\textbf{F}}_{\text{DL}}}\left( :,{{n}_{\text{d}}}+1 \right)={{\textbf{R}}_{\text{DL}}}\left( :,{{n}_{\text{d}}}+1 \right)+ {{\textbf{P}}_{\text{DL}}}\left( :,:,{{n}_{\text{d}}}+1 \right){{\boldsymbol{V}}_{m-1 }}$;\\
\quad\quad \textbf{end} \textbf{for}\\
\quad\quad \textbf{for} $n_\text{r}=0$ to $n_{\text{r}}^{\max }$ \textbf{do}\\
\quad\quad\quad ${{\textbf{F}}_{\text{RP}}}\left( :,n_\text{r}+1 \right)={{\textbf{R}}_{\text{RP}}}\left( :,n_\text{r}+1 \right)+ {{\textbf{P}}_{\text{RP}}}\left( :,:,n_\text{r}+1 \right){{\boldsymbol{V}}_{m-1 }}$;\\
\quad\quad \textbf{end} \textbf{for}\\
\quad\quad $\textbf{S}=\left[ {{\boldsymbol{f}}_{\text{DP}}},{{\textbf{F}}_{\text{DL}}},{{\textbf{F}}_{\text{RP}}} \right]$;\\
\quad\quad ${{\boldsymbol{V}}_{m }}=\max \left( \textbf{S},2 \right)$;\\
\quad\textbf{end} \textbf{for}
\STATE \textbf{Output}: ${{\boldsymbol{V}}_{M }}$, ${\boldsymbol{V}_{M}^{'}}$.
\end{algorithmic}
\end{algorithm}

\section{Simulation results}\label{Simulation results}
In this section, we investigate the on-time reception rate of a sequence of $M$ packets transmission over a Rayleigh fading channel.
    In particular, we transmit a sequence of $M$ packets over the channel and schedule each packet with the optimal scheduling policy obtained by Algorithm \ref{Alg: Optimal scheduling strategy algorithm}.
We then calculate the corresponding on-time reception rate (which is referred to as the \textit{simulation result}) by counting the packets received $\delta$-on-time.
    Moreover, we also calculate the on-time reception rate theoretically using Algorithm \ref{Total expected reward algorithm}, which is referred to as the \textit{theoretical results}.

    The distribution of  the channel power gain of the Rayleigh fading channel is given by
    \begin{equation}
        f_\gamma(x)= \lambda e^{-\lambda x}.
    \end{equation}

We set the channel parameter as $\lambda=2$, the transmit power of the source node as $P_{\text{t}}=1$ W, the distance between the source  and  destination nodes as $d=100$ m, the path loss exponent as $\alpha=2$, and the channel noise as $\sigma^2=10^{-4}$ W.
    For a given SNR threshold $V_{\text{T}}$, the probability that the transmitted packet can be successfully decoded by the destination node would be $p=\exp(-\lambda V_{\text{T}}d^\alpha \sigma^2) =\exp(-2V_{\text{T}})$ (see \eqref{rt:prb_suc}).
Thus, we can adjust the probability of successful transmissions by changing the threshold $V_{\text{T}}$, as shown in Fig. \ref{fig:p_vt}.
    For example, we have $p=0.2$ if  $V_{\text{T}}=0.8047$.

\begin{figure}[!t]
  \centering
  \includegraphics[width=3.7in]{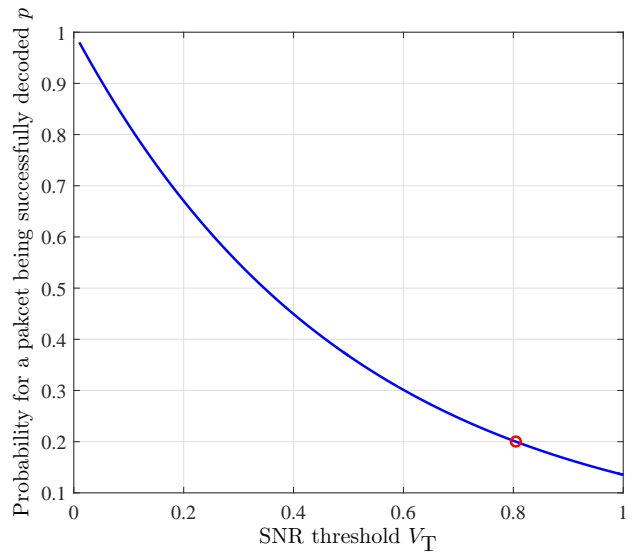}\\
  \caption{The probability for a packet being successfully received versus SNR threshold ($P_{\text{t}}=1$, $d=100$, $\alpha=2$, and $\sigma^2=10^{-4}$).}
  \label{fig:p_vt}
\end{figure}

 Without loss of generality, we consider a finite numbers of states and set the maximum and the minimum state as ${\iota }_{\max }=500$ and ${\iota }_{\min }=-500$, respectively, i.e., $s\in\left\{ -500,-499,\cdots,500 \right\}$.
 In the simulation, we also limit the delay time and the number of retransmissions by $n_{\text{d}}^{\max }=20$ and $n_{\text{r}}^{\max }=20$.
       In the implementation of the MDP algorithm, we set the discount factor as $\alpha=0.999$.

\begin{figure}[!t]
  \centering
  \includegraphics[width=3.7in]{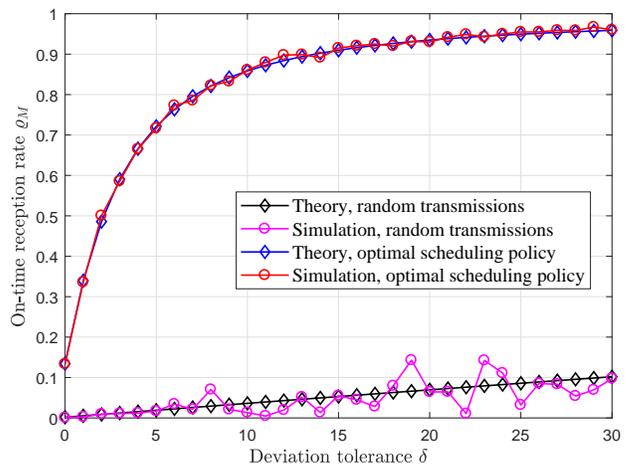}\\
  \caption{The on-time reception rate $\varrho_M$ versus the deviation tolerance $\delta$ ($p=0.2$, $M=10000$, and ${T}_{\text{tgt}}=5$).}
  \label{Receiving width}
\end{figure}

In Fig. \ref{Receiving width}, we investigate the behavior of the on-time reception rate $\varrho_M$ as a function of the deviation tolerance $\delta$.
    The probability of successful transmission is set to $p=0.2$, the number of packets is set to $M=10000$, and the target reception interval is set to $T_{\text{tgt}}=5$.
We observe that under the optimal packet scheduling policy obtained by Algorithm \ref{Alg: Optimal scheduling strategy algorithm}, the on-time reception rates  are much larger than that of the random transmission scheme.
    This shows that the proposed MDP based packet scheduling is very effective.
In case $\delta=0$, the $\delta$-on-time requirement reduces to the strictly on-time.
    From the figure, it can be seen that the corresponding on-time reception rates are relatively small, even though the optimal packet scheduling is used.
In fact, due to the fading property of the wireless channel, it is very difficult to alleviate the randomness of transmissions.
    Nevertheless, by using the optimal scheduling policy, the strictly on-time reception rate can be increased about 13\%, which is much larger than that of random transmissions.
Moreover, it is seen that our simulation results and theoretical results matches well.

Fig. \ref{the_interval_of_expected_receiving_time1} presents how the on-time reception  rate $\varrho_M$ changes with the target reception interval ${{T}_{\text{tgt}}}$.
    Besides the superiority of the optimal scheduling, we observe that the on-time reception rate increases with ${{T}_{\text{tgt}}}$.
This is because when  ${{T}_{\text{tgt}}}$ is relatively large, we have more freedom of scheduling.
    For the random transmissions, it is observed that $\varrho_M$ changes differently and reaches its maximum at ${{T}_{\text{tgt}}}=5$, which is exactly the expected value of the transmission time, i.e., $\mathbb{E}(S)={1}/{p}=5$.
This is in consistent with our intuitions that most of the transmission times fall into a finite range around their common expectation.

\begin{figure}[!t]
  \centering
  \includegraphics[width=3.7in]{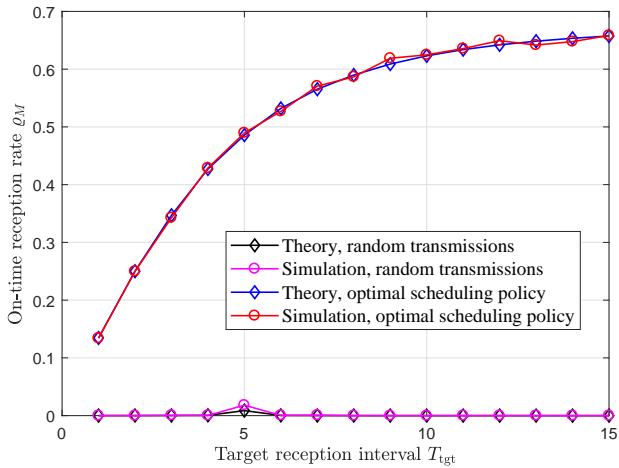}\\
  \caption{The on-time reception rate versus the target reception interval ${T}_{\text{tgt}}$ ($p=0.2$, $M=10000$, and $\delta=2$).}
  \label{the_interval_of_expected_receiving_time1}
\end{figure}

We plot the on-time reception rate $\varrho_M$ as a function of the successful reception probability $p$ (i.e., the reliability of the fading channel) in Fig. \ref{Receiving_probability2}, in which we set ${T}_{\text{tgt}}=4$, $M=10000$, and $\delta=2$.
     With the optimal scheduling policy, it is seen that $\varrho_M$ is increasing with $p$ and almost approaches the unity as $p$ reaches $0.5$.
Under the random transmission scheme, however, $\varrho_M$ does not change much as $p$ is increased.
    This is because when $p$ increases, although the variance (randomness) of the transmission time becomes smaller, the expected reception time of a packet deviates the target reception time more, unless ${T}_{\text{tgt}}={1}/{p}$ holds.

\begin{figure}[!t]
  \centering
  \includegraphics[width=3.7in]{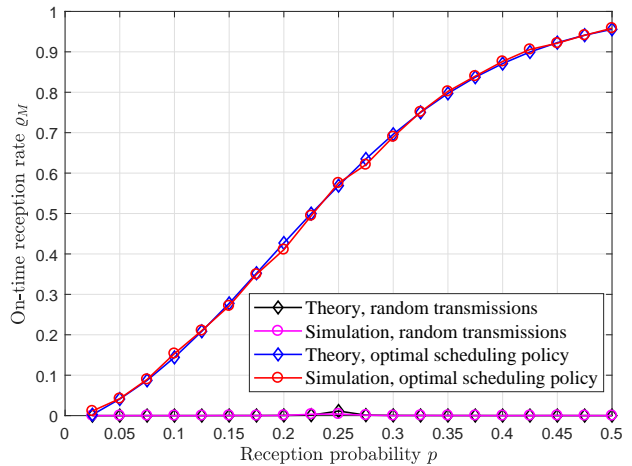}\\
  \caption{The on-time reception rate versus the reception probability $p$ (${T}_{\text{tgt}}=4$, $M=10000$, and $\delta=2$).}
  \label{Receiving_probability2}
\end{figure}

%
In Fig. \ref{CONVERGENCE_1}, we present how the on-time reception rate changes when the length $M$ of the packet sequence increases. For the random transmission strategy, we observe that the on-time reception rate decreases with $M$ and is expected to approach zero as $M$ goes to infinity. This is because the channel gains are random and difficult to predict while the accumulated deviation from the target times increases with $M$. For the transmission with optimal packet scheduling, it is seen that the on-time reception rate is much larger and converges to a constant as $M$ goes to infinity. By optimally scheduling the packets, however, the gain in the on-time reception rate is also limited, since the randomness of the channel cannot be removed completely.

\begin{figure}[!t]
  \centering
  \includegraphics[width=3.7in]{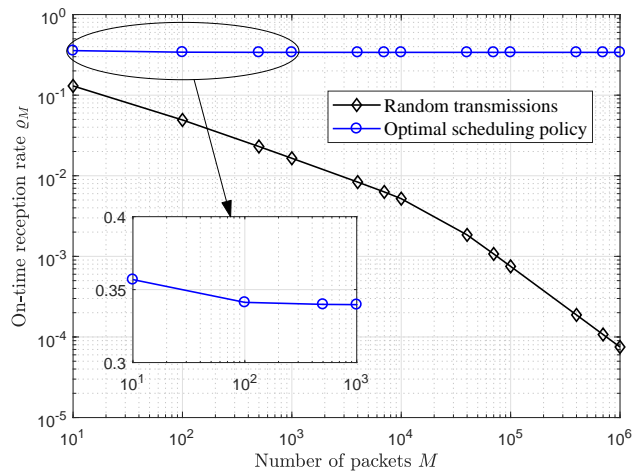}\\
  \caption{The on-time reception rate versus the number of packets $M$ ($p$=0.2, ${T}_{\text{tgt}}=5$ and $\delta=1$).}
  \label{CONVERGENCE_1}
\end{figure}

In Fig. \ref{pie}, we present how many packets are delayed, dropped, and repeated, respectively. As shown in the pie chart, we see that over 60\% of packets are repeated, which is because ${T}_{\text{tgt}}=33$ is relatively large.
    From the circle labeled curve, we also see that most of these packets are repeated by $5$ to $10$ times, since ${T}_{\text{tgt}}/{\mathbb{E}\left[ S  \right]}\;=p{{T}_{\text{tgt}}}=6.6$. Among the 36.3\% of delayed packets, most of them are delayed for $20$ slots, which is also because ${T}_\text{tgt}$ is large. In addition, only 1\% of packets need to be dropped in this setting.

\begin{figure}[!t]
  \centering
  \includegraphics[width=3.5in]{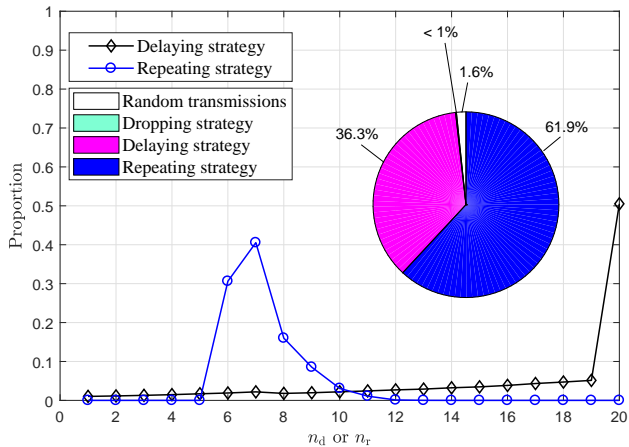}\\
  \caption{The proportions of the controlling strategies being used under the optimal scheduling policy ($p$=0.2, $M$=1000000, ${T}_{\text{tgt}}=33$ and $\delta=3$).}
  \label{pie}
\end{figure}

\section{Conclusion}\label{Conclusion}
In this paper, we have proposed an evaluation framework for the on-time communication over a fading channel. Due the fading property of the channel, the time to successfully deliver a packet is random so that it is very difficult to receive a packet in an expected slot. Thus, we increased the on-time reception rate of the packets significantly by optimally delaying, dropping, and repeating some of them. However, the improvement is also limited and the 100\% on-time transmission can never be achieved, unless the randomness in the channel gains can be completely removed (e.g., can be fully predicted). In our future work, we shall further combat the randomness of fading channels by power allocations, variable-rate compressions, and parallel transmissions. By  optimizing the $\delta$-on-time reception rate jointly, it is expected that TSN networks can be implemented over wireless networks in the near future.

\appendix
\subsection{Proof of Proposition \ref{prop:pxm}}
\begin{proof}\label{proof of corollary 1}
We denote the transmission time of the $m$-th packet as ${{\tau}_{m}}$ and the probability that the $m$-th packet is received $\delta$-on time as $P\left( {{x}_{m}} \right)$. For the $m$-th packet, we have
\begin{align}
P\left( {{x}_{m}} \right)=\Pr \left\{ m{{T}_{\text{tgt}}}-\delta \le \sum\limits_{k=1}^{m}{{{\tau}_{k}}}\le m{{T}_{\text{tgt}}}+\delta  \right\}.
\end{align}
Since transmission time ${\tau}_{k}$ follows the geometric distribution (cf. \eqref{eq:service time}) and the total transmission time $\sum_{k=1}^m{\tau}_k$ follows the negative binomial distribution with parameter $p$, we have
\begin{align}
\Pr \left\{ \sum\limits_{k=1}^{m}{{{\tau }_{k}}}=j \right\}=C_{j-1}^{m-1}{{p}^{m}}{{\left( 1-p \right)}^{j-m}},j=m,m+1,\ldots
\end{align}
To calculate $P\left( {{x}_{m}} \right)$, we consider the following two cases.
\subsubsection{$m{T}_{\rm{tgt}}\le m+\delta$}
Since ${{\tau}_{m}}\ge 1$, we have $\sum\limits_{k=1}^{m}{{{\tau}_{k}}}\ge m$, and $m{{T}_{\text{tgt}}}-\delta \le m\le \sum\limits_{k=1}^{m}{{{\tau }_{k}}}$. Thus,
\begin{align}
  P\left( {{x}_{m}} \right)& =\Pr \left\{ \sum\limits_{k=1}^{m}{{{\tau}_{k}}}\le m{T}_{\text{tgt}}+\delta  \right\}\nonumber \\
& ={{p}^{m}}+\ldots +C_{m{{T}_{\text{tgt}}}+\delta -1}^{m-1}{{p}^{m}}{{\left( 1-p \right)}^{m{{T}_{\text{tgt}}}+\delta -m}}\nonumber \\
& =\sum\limits_{k=m}^{m{{T}_{\text{tgt}}}+\delta }{C_{k-1}^{m-1}}{{p}^{m}}{{\left( 1-p \right)}^{k-m}}.
\end{align}
\subsubsection{$m{T}_{\rm{tgt}}>m+\delta$}
In this case, we have
\begin{align}
 P\left( {{x}_{m}} \right)& =\Pr \left\{ \sum\limits_{k=1}^{m}{{{\tau}_{k}}}\le m{{T}_{\text{tgt}}}+\delta  \right\}\nonumber \nonumber  \\
 &\quad\quad\quad\qquad\quad \ -\Pr \left\{ \sum\limits_{k=1}^{m}{{{\tau}_{k}}}\le m{{T}_{\text{tgt}}}-\delta -1 \right\}\nonumber \\
 & =\sum\limits_{k=m{{T}_{\text{tgt}}}-\delta }^{m{{T}_{\text{tgt}}}+\delta }{C_{k-1}^{m-1}{{p}^{m}}{{\left( 1-p \right)}^{k-m}}}.
\end{align}
This completes the proof of Proposition \ref{prop:pxm}.
\end{proof}
\subsection{Proof of Theorem \ref{classical}}
\begin{proof}\label{Proof of Theorem 1}
Let $P\left( {{x}_{k}} \right)$ and $P\left( \overline{{{x}_{k}}} \right)$ be the probability for the $k$-th packet to be and to be not received $\delta$-on time, respectively. We denote the probability that $k$ packets out of the $M$ packets are received $\delta$-on time as $P\left( x_{k}^{M} \right)$. Under this setting, We prove the theorem by mathematical induction.

We start from $M=1$ and readily see that $P\left( x_{1}^{1} \right)=P\left( {{x}_{1}} \right)$. We assume that Theorem \ref{classical} holds for $M=n$, i.e.,
\begin{align}\label{eq:Proof of classical probability}
\sum\limits_{k=1}^{M}{kP\left( x_{k}^{M} \right)}=\sum\limits_{k=1}^{M}{P\left( {{x}_{k}} \right)}.
\end{align}
That is,
\begin{align}\label{eq:c}
 & nP\left( {{x}_{1}}{{x}_{2}}\cdots {{x}_{n}} \right)\nonumber \\
 & +\left( n-1 \right)\left[ P\left( \overline{{{x}_{1}}}{{x}_{2}}\cdots {{x}_{n}} \right)+\ldots +P\left( {{x}_{1}}\cdots {{x}_{n-1}}\overline{{{x}_{n}}} \right) \right]+\nonumber  \\
 & \left( n-2 \right)\left[ P\left( \overline{{{x}_{1}}}\ \overline{{{x}_{2}}}{x}_{3}\cdots {{x}_{n}} \right)+\ldots+P\left( {{x}_{1}}\cdots {x}_{n-2}\overline{{{x}_{n-1}}}\ \overline{{{x}_{n}}} \right) \right]\nonumber  \\
 & +\ldots +\left[ P\left( {{x}_{1}}\overline{{{x}_{2}}}\cdots \overline{{{x}_{n}}} \right)+\ldots +P\left( \overline{{{x}_{1}}}\ \overline{{{x}_{2}}}\cdots {{x}_{n}} \right) \right]\nonumber  \\
 & =P\left( {{x}_{1}} \right)+P\left( {{x}_{2}} \right)+\ldots +P\left( {{x}_{n}} \right).
\end{align}
For $M=n+1$, we then have
\begin{align}\label{eq:Proof of classical probability litter}
  & \sum\limits_{k=1}^{n+1}{kP\left(x_{k}^{n+1} \right)}\nonumber \\
  &\stackrel{(a)}{=}\left( n+1 \right)P\left( {{x}_{1}}{{x}_{2}}\cdots {{x}_{n}}{{x}_{n+1}} \right)\nonumber \\
  & +n\left[ P\left( \overline{{{x}_{1}}}{{x}_{2}}\cdots {{x}_{n}}{{x}_{n+1}} \right)+\ldots +P\left( {{x}_{1}}{{x}_{2}}\cdots {{x}_{n}}\ \overline{{{x}_{n+1}}} \right) \right]\nonumber \\
  & +\left( n-1 \right)\left[ P\left( \overline{{{x}_{1}}}\ \overline{{{x}_{2}}}\cdots {{x}_{n+1}} \right)+\ldots +P\left( {{x}_{1}}\cdots \overline{{{x}_{n}}}\ \overline{{{x}_{n+1}}} \right) \right]\nonumber \\
  & +\ldots \left[ P\left( {{x}_{1}}\overline{{{x}_{2}}}\ \overline{{{x}_{3}}}\cdots \overline{{{x}_{n+1}}} \right)+\ldots +P\left( \overline{{{x}_{1}}}\ \overline{{{x}_{2}}}\cdots {{x}_{n+1}} \right) \right]\nonumber \\
  &\stackrel{(b)}{=}nP\left( {{x}_{1}}{{x}_{2}}\cdots {{x}_{n}} \right)\nonumber \\
  &+\left( n-1 \right)\left[ P\left( \overline{{{x}_{1}}}{{x}_{2}}\cdots {{x}_{n}} \right)+\ldots +P\left( {{x}_{1}}{{x}_{2}}\cdots \overline{{{x}_{n}}} \right) \right]\nonumber  \\
  &+\left( n-2 \right)\left[ P\left( \overline{{{x}_{1}}}\ \overline{{{x}_{2}}}\cdots {{x}_{n}} \right)+\ldots +P\left( {{x}_{1}}{{x}_{2}}\cdots \overline{{{x}_{n-1}}}\ \overline{{{x}_{n}}} \right) \right]\nonumber  \\
  &+\ldots \left[ P\left( {{x}_{1}}\overline{{{x}_{2}}}\ \overline{{{x}_{3}}}\cdots \overline{{{x}_{n}}} \right)+\ldots +P\left( \overline{{{x}_{1}}}\ \overline{{{x}_{2}}}\cdots {{x}_{n}} \right) \right]\nonumber  \\
  &+\left[ P\left( {{x}_{1}}{{x}_{2}}\cdots {{x}_{n+1}} \right)+\ldots +P\left( \overline{{{x}_{1}}}\ \overline{{{x}_{2}}}\ \overline{{{x}_{3}}}\cdots {{x}_{n+1}} \right) \right]\nonumber \\
  &\stackrel{(c)}{=}P\left( {{x}_{1}} \right)+P\left( {{x}_{2}} \right)+\ldots +P\left( {{x}_{n+1}} \right)\nonumber  \\
  & =\sum\limits_{k=1}^{n+1}{P\left( {{x}_{k}} \right)},
\end{align}
in which $\left( c \right)$ follows from \eqref{eq:c} and $\left( b \right)$ is obtained by re-organizing the equation $\left( a \right)$. For example, the first term of \eqref{eq:Proof of classical probability litter} can be calculated by
\begin{align}
& \left( n+1 \right)P\left( {{x}_{1}}\cdots{{x}_{n+1}} \right)+nP\left( {{x}_{1}}\cdots {{x}_{n}}\overline{{{x}_{n+1}}} \right) \nonumber \\
& =nP\left( {{x}_{1}}{{x}_{2}}\cdots {{x}_{n}} \right)+P\left( {{x}_{1}}{{x}_{2}}\cdots {{x}_{n}}{{x}_{n+1}} \right).
\end{align}
Thus, \eqref{eq:expectation} holds true for $M=n+1$, and thus holds for all $M\ge1$ and the proof of Theorem \ref{classical} is completed.
\end{proof}
As an illustrative example, we have
\begin{align}
 & \sum\limits_{k=1}^{3}{kP\left( x_{k}^{3} \right)}\nonumber \\
 & =3P\left( {{x}_{1}}{{x}_{2}}{{x}_{3}} \right)\nonumber \\
 & +2\left[ P\left( {{x}_{1}}{{x}_{2}}\overline{{{x}_{3}}} \right)+P\left( {{x}_{1}}\overline{{{x}_{2}}}{{x}_{3}} \right)+P\left( \overline{{{x}_{1}}}{{x}_{2}}{{x}_{3}} \right) \right]\nonumber \\
 & +P\left( {{x}_{1}}\overline{{{x}_{2}}}\ \overline{{{x}_{3}}} \right)+P\left( \overline{{{x}_{1}}}\ \overline{{{x}_{2}}}{{x}_{3}} \right)+P\left( \overline{{{x}_{1}}}{{x}_{2}}\overline{{{x}_{3}}} \right)\nonumber \\
 & =2P\left( {{x}_{1}}{{x}_{2}} \right)+P\left( {{x}_{1}}\overline{{{x}_{2}}} \right)+P\left( \overline{{{x}_{1}}}{{x}_{2}} \right)\nonumber \\
 & +\left[ P\left( {{x}_{1}}{{x}_{2}}{{x}_{3}} \right)+P\left( {{x}_{1}}\overline{{{x}_{2}}}{{x}_{3}} \right)+P\left( \overline{{{x}_{1}}}{{x}_{2}}{{x}_{3}} \right)+P\left( \overline{{{x}_{1}}}\ \overline{{{x}_{2}}}{{x}_{3}} \right) \right]\nonumber \\
 & =P\left( {{x}_{1}} \right)+P\left( {{x}_{2}} \right)+P\left( {{x}_{3}} \right)\nonumber \\
 & =\sum\limits_{k=1}^{3}{P\left( {{x}_{k}} \right)}
\end{align}
\subsection{Proof of Proposition \ref{Probability distribution function of repetition strategy}}
\begin{proof}\label{proof of corollary 2}
Before proving Proposition \ref{Probability distribution function of repetition strategy}, we first prove the following equation
\begin{align}\label{eq:Combination number}
\sum\limits_{{{y}_{1}}=1}^{z-m}{\sum\limits_{{{y}_{2}}=1}^{z-{{y}_{1}}-\left( m-1 \right)}{\ldots \sum\limits_{{{y}_{m}}=1}^{z-\sum\limits_{k=1}^{m-1}{{{y}_{k}}}-1}{1}}}=C_{z-1}^{m},
\end{align}
in which $z\ge m+1$, $m\ge 1$.

We prove \eqref{eq:Combination number} using the mathematical induction. For $m=1$, it is clear that the equation is true. Suppose \eqref{eq:Combination number} holds for $m=n$ and we have
\begin{align}
\sum\limits_{{{y}_{1}}=1}^{z-n}{\sum\limits_{{{y}_{2}}=1}^{z-{{y}_{1}}-\left( n-1 \right)}{\ldots \sum\limits_{{{y}_{n}}=1}^{z-\sum\limits_{k=1}^{n-1}{{{y}_{k}}}-1}{1}}}=C_{z-1}^{n}.
\end{align}
For $m=n+1$, we then have
\begin{align}
  & \sum\limits_{{{y}_{1}}=1}^{z-\left( n+1 \right)}{\sum\limits_{{{y}_{2}}=1}^{z-{{y}_{1}}-n}{\ldots \sum\limits_{{{y}_{n}}=1}^{z-\sum\limits_{k=1}^{n-1}{{{y}_{k}}}-2}{\sum\limits_{{{y}_{n+1}}=1}^{z-\sum\limits_{k=1}^{n}{{{y}_{k}}}-1}{1}}}}\nonumber \\
  & =\sum\limits_{{{y}_{1}}=1}^{z-n-1}{\left( \sum\limits_{{{y}_{2}}=1}^{z-{{y}_{1}}-n}{\ldots \sum\limits_{{{y}_{n}}=1}^{z-\sum\limits_{k=1}^{n-1}{{{y}_{k}}}-2}{\sum\limits_{{{y}_{n+1}}=1}^{z-\sum\limits_{k=1}^{n}{{{y}_{k}}}-1}{1}}} \right)}\nonumber \\
  & =\sum\limits_{y=1}^{z-n-1}{\left( \sum\limits_{{{y}_{1}}=1}^{z-y-n}{\ldots \sum\limits_{{{y}_{n}}=1}^{z-y-\sum\limits_{k=1}^{n-1}{{{y}_{k}}}-1}{1}} \right)}\nonumber \\
  & =\sum\limits_{y=1}^{z-n-1}{C_{z-y-1}^{n}}\nonumber \\
  & =C_{z-2}^{n}+C_{z-3}^{n}+\ldots +C_{n}^{n}\nonumber \\
  & =C_{z-2}^{n}+C_{z-2}^{n+1}-C_{z-3}^{n+1}+C_{z-3}^{n+1}\ldots -C_{n+1}^{n+1}+C_{n}^{n}\nonumber \\
  & =C_{z-2}^{n}+C_{z-2}^{n+1}\nonumber \\
  & =C_{z-1}^{n+1}.
\end{align}
That is, \eqref{eq:Combination number} also holds for $m=n+1$. Therefore, \eqref{eq:Combination number} holds for all $m\ge 1$.

We denote the transmission time of a packet under the repeat strategy and at most $n_\text{r}$ retransmissions as ${S }_{n_\text{r}}$, the transmission time of the $i$-th retransmission as ${s}_{\text{r},i}$, the transmission time of the first transmission as ${s}_{\text{r},0}$. When a packet is received before the target reception range, the packet will be retransmitted and we have
\begin{align}\label{Proof method and formula of probability distribution}
  & \Pr \left\{ {{S}_{n_\text{r}}}-{{T}_{\text{tgt}}}>j \right\}\nonumber \\
  & =\Pr \left\{ \sum\limits_{i=0}^{n_\text{r}}{{{s}_{\text{r},i}}-{T}_{\text{tgt}}>j,\sum\limits_{i=0}^{n_\text{r}-1}{{{s}_{\text{r},i}}}<{T}_{\text{tgt}}-\delta } \right\}\nonumber  \\
  &\quad+\Pr \left\{ \sum\limits_{i=0}^{n_\text{r}-1}{{s}_{\text{r},i}}-{T}_{\text{tgt}}>j,\sum\limits_{i=0}^{n_\text{r}-2}{{{s}_{\text{r},i}}}<{{T}_{\text{tgt}}}-\delta \le \sum\limits_{i=0}^{n_\text{r}-1}{{{s}_{\text{r},i}}}  \right\}\nonumber \\
  &\quad+\ldots +\Pr \left\{ {{s}_{0}}-{T}_{\text{tgt}}>j,{{s}_{0}}\ge {T}_{\text{tgt}}-\delta  \right\}.
\end{align}
\subsubsection{${T}_{\rm{tgt}}\ge 1+n_\text{r}+\delta$}
When $j\ge -1-\delta$, we have $j+{T}_{\text{tgt}}\ge {T}_{\text{tgt}}-\delta -1$ and
\begin{align}
  & \Pr \left\{ {{S}_{n_\text{r}}}-{T}_{\text{tgt}}>j \right\}\nonumber  \\
  & =\Pr \left\{ \sum\limits_{i=0}^{n_\text{r}}{{{s}_{\text{r},i}}-{T}_{\text{tgt}}>j,\sum\limits_{i=0}^{n_\text{r}-1}{{{s}_{\text{r},i}}}<{T}_{\text{tgt}}-\delta } \right\}\nonumber  \\
  & \quad\quad\quad\quad+\Pr \left\{ \sum\limits_{i=0}^{n_\text{r}-1}{{{s}_{\text{r},i}}}-{T}_{\text{tgt}}>j,\sum\limits_{i=0}^{n_\text{r}-2}{{{s}_{\text{r},i}}}<{T}_{\text{tgt}}-\delta  \right\}\nonumber  \\
  & \quad\quad\quad\quad+\ldots +\Pr \left\{ {{s}_{0}}-{T}_{\text{tgt}}>j \right\}\nonumber  \\
  & =\Pr \left\{ \sum\limits_{i=0}^{n_\text{r}}{{{s}_{\text{r},i}}-{T}_{\text{tgt}}>j,\sum\limits_{i=0}^{n_\text{r}-1}{{{s}_{\text{r},i}}}<{T}_{\text{tgt}}-\delta } \right\}\nonumber  \\
  & \quad\quad\quad\quad+\Pr \left\{ {{S}_{n_\text{r}-1}}-{T}_{\text{tgt}}>j \right\}\nonumber  \\
  & =\sum\limits_{{{y}_{0}}=1}^{{T}_{\text{tgt}}-\delta -n_\text{r}}\sum\limits_{{{y}_{1}}=1}^{{T}_{\text{tgt}}-\delta -{{y}_{0}}- n_\text{r}+1}{\cdots }\sum\limits_{{{y}_{n_\text{r}-1}}=1}^{{T}_{\text{tgt}}-\delta -\sum\limits_{i=0}^{n_\text{r}-2}{{{y}_{i}}}-1}V\nonumber  \\
  & \quad\quad\quad\quad+\Pr \left\{ {{S}_{n_\text{r}-1}}-{T}_{\text{tgt}}>j \right\},
\end{align}
in which
\begin{align}\label{eq:V}
  V& ={{p}^{n_\text{r}}}{{\left( 1-p \right)}^{\sum\limits_{i=0}^{n_\text{r}-1}{{{y}_{i}}}-n_\text{r}}}\nonumber  \\
   &\qquad\qquad\qquad \cdot \Pr \left\{ {{s}_{\text{r},n_\text{r}}}>j+{T}_{\text{tgt}}-\sum\limits_{i=0}^{n_\text{r}-1}{{{y}_{i}}},{{s}_{\text{r},i}}={{y}_{i}} \right\}\nonumber  \\
   & =\left\{ \begin{aligned}
   & {{p}^{n_\text{r}}}{{\left( 1-p \right)}^{j+{{T}_{\text{tgt}}}-n_\text{r}}},&& j\ge \sum\limits_{i=0}^{n_\text{r}-1}{{{y}_{i}}}-{T}_{\text{tgt}} \\
   & {{p}^{n_\text{r}}}{{\left( 1-p \right)}^{\sum\limits_{i=0}^{n_\text{r}-1}{{{y}_{i}}}-n_\text{r}}},&& j<\sum\limits_{i=0}^{n_\text{r}-1}{{{y}_{i}}}-{T}_{\text{tgt}}. \\
\end{aligned} \right.
\end{align}
Since $j\ge -1-\delta$, $\sum\limits_{i=0}^{n_\text{r}-1}{{{s}_{\text{r},i}}}=\sum\limits_{i=1}^{n_\text{r}-1}{{{y}_{i}}}<{T}_{\text{tgt}}-\delta$, we have $j\ge \sum\limits_{i=0}^{n_\text{r}-1}{{{y}_{i}}}-{T}_{\text{tgt}}$, and then $V={{p}^{n_\text{r}}}{{\left( 1-p \right)}^{j+{{T}_{\text{tgt}}}-n_\text{r}}}$. From \eqref{eq:Combination number}, we have
\begin{align}
  & \Pr \left\{ {{S}_{n_\text{r}}}-{T}_{\text{tgt}}>j \right\}\nonumber  \\
  & =\sum\limits_{{{y}_{0}}=1}^{{T}_{\text{tgt}}-\delta -n_\text{r}}{\sum\limits_{{{y}_{1}}=1}^{{T}_{\text{tgt}}-\delta -{{y}_{0}}-\left( n_\text{r}-1 \right)}{\cdots }\sum\limits_{{{y}_{n_\text{r}-1}}=1}^{{T}_{\text{tgt}}-\delta -\sum\limits_{i=0}^{n_\text{r}-2}{{{y}_{i}}}-1}V}\nonumber  \\
  &\quad\ \ +\Pr \left\{ {{S}_{n_\text{r}-1}}-{T}_{\text{tgt}}>j \right\},\nonumber  \\
  & =C_{{T}_{\text{tgt}}-\delta -1}^{n_\text{r}}{{p}^{n_\text{r}}}{{\left( 1-p \right)}^{j+{T}_{\text{tgt}}-n_\text{r}}}+\Pr \left\{ {{S}_{n_\text{r}-1}}-{T}_{\text{tgt}}>j \right\}\nonumber  \\
  & =C_{{T}_{\text{tgt}}-\delta -1}^{n_\text{r}}{{p}^{n_\text{r}}}{{\left( 1-p \right)}^{j+{T}_{\text{tgt}}-n_\text{r}}}\nonumber \\
  & \quad\ \ +C_{{T}_{\text{tgt}}-\delta -1}^{n_\text{r}-1}{{p}^{n_\text{r}-1}}{{\left( 1-p \right)}^{j+{T}_{\text{tgt}}-n_\text{r}+1}}+\ldots +{{\left( 1-p \right)}^{j+{T}_{\text{tgt}}}}\nonumber  \\
  & =\sum\limits_{i=0}^{n_\text{r}}{C_{{T}_{\text{tgt}}-\delta -1}^{i}{{p}^{i}}{{\left( 1-p \right)}^{j+{T}_{\text{tgt}}-i}}}.
\end{align}

In case $n_\text{r}-{T}_{\text{tgt}}\le j<-1-\delta$, we have \eqref{eq:appendix_kualan_0} from $\eqref{Proof method and formula of probability distribution}$, and we have \eqref{eq:appendix_kualan_1} from $\eqref{eq:V}$.
\begin{figure*}[b]
\hrule
{\small
\begin{align}\label{eq:appendix_kualan_0}
  & \Pr \left\{ {{S}_{n_\text{r}}}-{{T}_{\text{tgt}}}>j \right\}\nonumber \\
 & =\Pr \left\{ \sum\limits_{i=0}^{n_\text{r}}{{{s}_{\text{r},i}}}-{{T}_{\text{tgt}}}>j,\sum\limits_{i=0}^{n_\text{r}-1}{{{s}_{\text{r},i}}}<{{T}_{\text{tgt}}}-\delta  \right\}+\Pr \left\{ \sum\limits_{i=0}^{n_\text{r}-1}{{{s}_{\text{r},i}}}\ge {{T}_{\text{tgt}}}-\delta ,\sum\limits_{i=0}^{n_\text{r}-2}{{{s}_{\text{r},i}}}<{{T}_{\text{tgt}}}-\delta  \right\}+\ldots +\Pr \left\{ {{s}_{\text{r},0}}\ge {{T}_{\text{tgt}}}-\delta  \right\}\nonumber \\
 & =\Pr \left\{ \sum\limits_{i=0}^{n_\text{r}}{{{s}_{\text{r},i}}}-j-\delta >{{T}_{\text{tgt}}}-\delta ,\sum\limits_{i=0}^{n_\text{r}-1}{{{s}_{\text{r},i}}}<{{T}_{\text{tgt}}}-\delta  \right\}+\Pr \left\{ \sum\limits_{i=0}^{n_\text{r}-1}{{{s}_{\text{r},i}}}\ge {{T}_{\text{tgt}}}-\delta ,\sum\limits_{i=0}^{n_\text{r}-2}{{{s}_{\text{r},i}}}<{{T}_{\text{tgt}}}-\delta  \right\}\nonumber \\
 & \qquad\qquad\qquad\qquad\qquad\qquad\qquad\qquad\qquad\qquad\qquad\qquad\qquad\qquad\qquad\qquad\qquad\qquad\quad\ \ \ \ \ \ \ \ +\ldots +\Pr \left\{ {{s}_{\text{r},0}}\ge {{T}_{\text{tgt}}}-\delta  \right\}\nonumber \\
 & =\Pr \left\{ \sum\limits_{i=0}^{n_\text{r}}{{{s}_{\text{r},i}}}-j-\delta >{{T}_{\text{tgt}}}-\delta  \right\}\nonumber \\
 & =\Pr \left\{ \sum\limits_{i=0}^{n_\text{r}}{{{s}_{\text{r},i}}}-{{T}_{\text{tgt}}}>j \right\}.
\end{align}
}
\end{figure*}
Finally, we have \eqref{eq:appendix_kualan_2}.

\begin{figure*}[b]
\hrule
\begin{align}\label{eq:appendix_kualan_1}
  &\Pr \left\{ \sum\limits_{i=0}^{n_\text{r}}{{{s}_{\text{r},i}}-{T}_{\text{tgt}}>j,\sum\limits_{i=0}^{n_\text{r}-1}{{{s}_{\text{r},i}}}<{T}_{\text{tgt}}-\delta } \right\}\nonumber  \\
  & =  \Pr \left\{ \sum\limits_{i=0}^{n_\text{r}}{{{s}_{\text{r},i}}-{T}_{\text{tgt}}>j,\sum\limits_{i=0}^{n_\text{r}-1}{{{s}_{\text{r},i}}}<{T}_{\text{tgt}}-\delta ,j\ge \sum\limits_{i=0}^{n_\text{r}-1}{{{s}_{\text{r},i}}}-{T}_{\text{tgt}}} \right\}\nonumber \\
  & +\Pr \left\{ \sum\limits_{i=0}^{n_\text{r}}{{{s}_{\text{r},i}}-{T}_{\text{tgt}}>j,\sum\limits_{i=0}^{n_\text{r}-1}{{{s}_{\text{r},i}}}<{T}_{\text{tgt}}-\delta ,j<\sum\limits_{i=0}^{n_\text{r}-1}{{{s}_{\text{r},i}}}-{T}_{\text{tgt}}} \right\}\nonumber  \\
  & =\Pr \left\{ \sum\limits_{i=0}^{n_\text{r}}{{{s}_{\text{r},i}}-{T}_{\text{tgt}}>j,j\ge \sum\limits_{i=0}^{n_\text{r}-1}{{{s}_{\text{r},i}}}-{T}_{\text{tgt}}} \right\}
   +\Pr \left\{ \sum\limits_{i=0}^{n_\text{r}-1}{{{s}_{\text{r},i}}}<{T}_{\text{tgt}}-\delta ,j<\sum\limits_{i=0}^{n_\text{r}-1}{{{s}_{\text{r},i}}}-{T}_{\text{tgt}} \right\}\nonumber  \\
  & = \sum\limits_{{{y}_{0}}=1}^{j+{T}_{\text{tgt}}-n_\text{r}+1}{\sum\limits_{{{y}_{1}}=1}^{j+{T}_{\text{tgt}}-{{y}_{0}}-n_\text{r}+2}{\cdots }\sum\limits_{{{y}_{n_\text{r}-1}}=1}^{j+{T}_{\text{tgt}}-\sum\limits_{i=0}^{n_\text{r}-2}{{{y}_{i}}}}V}
   +\Pr \left\{ \sum\limits_{i=0}^{n_\text{r}-1}{{{s}_{\text{r},i}}}<{T}_{\text{tgt}}-\delta ,j<\sum\limits_{i=0}^{n_\text{r}-1}{{{s}_{\text{r},i}}}-{T}_{\text{tgt}} \right\}\nonumber \\
  & = \sum\limits_{{{y}_{0}}=1}^{j+{T}_{\text{tgt}}-\left( n_\text{r}-1 \right)}{\sum\limits_{{{y}_{1}}=1}^{j+{T}_{\text{tgt}}-{{y}_{0}}-\left( n_\text{r}-2 \right)}{\cdots }\sum\limits_{{{y}_{n_\text{r}-1}}=1}^{j+{T}_{\text{tgt}}-\sum\limits_{i=0}^{n_\text{r}-2}{{{y}_{i}}}}{{{p}^{n_\text{r}}}{{\left( 1-p \right)}^{j+{T}_{\text{tgt}}-n_\text{r}}}}}
   +\Pr \left\{ \sum\limits_{i=0}^{n_\text{r}-1}{{{s}_{\text{r},i}}}<{T}_{\text{tgt}}-\delta ,j<\sum\limits_{i=0}^{n_\text{r}-1}{{{s}_{\text{r},i}}}-{T}_{\text{tgt}} \right\}\nonumber  \\
  & =C_{j+{T}_{\text{tgt}}}^{n_\text{r}}{{p}^{n_\text{r}}}{{\left( 1-p \right)}^{j+{T}_{\text{tgt}}-n_\text{r}}}
   +\Pr \left\{ \sum\limits_{i=0}^{n_\text{r}-1}{{{s}_{\text{r},i}}}<{T}_{\text{tgt}}-\delta  \right\}-\Pr \left\{ \sum\limits_{i=0}^{n_\text{r}-1}{{{s}_{\text{r},i}}}-{T}_{\text{tgt}}\le j \right\}.
\end{align}
\end{figure*}
\begin{figure*}
\hrule
\begin{align}\label{eq:appendix_kualan_2}
  & \Pr \left\{ {{S}_{n_\text{r}}}-{T}_{\text{tgt}}>j \right\}\nonumber  \\
  & =\Pr \left\{ \sum\limits_{i=0}^{n_\text{r}}{{s}_{\text{r},i}}-{T}_{\text{tgt}}>j,\sum\limits_{i=0}^{n_\text{r}-1}{{s}_{\text{r},i}}<{T}_{\text{tgt}}-\delta  \right\}+
    \Pr \left\{ \sum\limits_{i=0}^{n_\text{r}-1}{{s}_{\text{r},i}}\ge {T}_{\text{tgt}}-\delta ,\sum\limits_{i=0}^{n_\text{r}-2}{{s}_{\text{r},i}}<{T}_{\text{tgt}}-\delta  \right\}
    +\ldots +\Pr \left\{ {s}_{\text{r},0}\ge {T}_{\text{tgt}}-\delta  \right\}\nonumber  \\\nonumber  \\
  & =C_{j+{T}_{\text{tgt}}}^{n_\text{r}}{{p}^{n_\text{r}}}{{\left( 1-p \right)}^{j+{T}_{\text{tgt}}-n_\text{r}}}
    +\Pr \left\{ \sum\limits_{i=0}^{n_\text{r}-1}{{s}_{\text{r},i}}<{T}_{\text{tgt}}-\delta  \right\}-\Pr \left\{ \sum\limits_{i=0}^{n_\text{r}-1}{{s}_{\text{r},i}}-{T}_{\text{tgt}}\le j \right\}\nonumber \\
  &\quad +\Pr \left\{ \sum\limits_{i=0}^{n_\text{r}-1}{{s}_{\text{r},i}}\ge {T}_{\text{tgt}}-\delta ,\sum\limits_{i=0}^{n_\text{r}-2}{{s}_{\text{r},i}}<{T}_{\text{tgt}}-\delta  \right\}+\ldots +\Pr \left\{ {s}_{\text{r},0}\ge {T}_{\text{tgt}}-\delta  \right\}\nonumber  \\
  & =C_{j+{T}_{\text{tgt}}}^{n_\text{r}}{{p}^{n_\text{r}}}{{\left( 1-p \right)}^{j+{T}_{\text{tgt}}-n_\text{r}}}
    +1-\Pr \left\{ \sum\limits_{i=0}^{n_\text{r}-1}{{s}_{\text{r},i}}-{T}_{\text{tgt}}\le j \right\}\nonumber  \\
  & =C_{j+{T}_{\text{tgt}}}^{n_\text{r}}{{p}^{n_\text{r}}}{{\left( 1-p \right)}^{j+{T}_{\text{tgt}}-n_\text{r}}}+\Pr \left\{ \sum\limits_{i=0}^{n_\text{r}-1}{{s}_{\text{r},i}}-{{T}_{\text{tgt}}}>j \right\}\nonumber  \\
  & =C_{j+{T}_{\text{tgt}}}^{n_\text{r}}{{p}^{n_\text{r}}}{{\left( 1-p \right)}^{j+{T}_{\text{tgt}}-n_\text{r}}}+\Pr \left\{ {{S}_{n_\text{r}-1}}-{{T}_{\text{tgt}}}>j \right\}\nonumber  \\
  & =C_{j+{T}_{\text{tgt}}}^{n_\text{r}}{{p}^{n_\text{r}}}{{\left( 1-p \right)}^{j+{T}_{\text{tgt}}-n_\text{r}}}
    +C_{j+{T}_{\text{tgt}}}^{n_\text{r}-1}{{p}^{n_\text{r}-1}}{{\left( 1-p \right)}^{j+{T}_{\text{tgt}}-n_\text{r}+1}}+\ldots +{{\left( 1-p \right)}^{j+{T}_{\text{tgt}}}}\nonumber  \\
  & =\sum\limits_{i=0}^{n_\text{r}}{C_{j+{T}_{\text{tgt}}}^{i}{{p}^{i}}{{\left( 1-p \right)}^{j+{T}_{\text{tgt}}-i}}}.
\end{align}
\end{figure*}

In case $j<n_\text{r}-{T}_{\text{tgt}}$, we have $j+{T}_{\text{tgt}}<n_\text{r}\le {T}_{\text{tgt}}-\delta -1$. From $\eqref{eq:V}$, we obtain
\begin{align}
  & \Pr \left\{ {{S}_{n_\text{r}}}-{T}_{\text{tgt}}>j \right\}\nonumber \\
  & =\Pr \left\{ \sum\limits_{i=0}^{n_\text{r}-1}{{{s}_{\text{r},i}}}<{T}_{\text{tgt}}-\delta  \right\} \nonumber  \\
  &\quad\quad\quad\quad\quad +\Pr \left\{ \sum\limits_{i=0}^{n_\text{r}-1}{{{s}_{\text{r},i}}}\ge {T}_{\text{tgt}}-\delta ,\sum\limits_{i=0}^{n_\text{r}-2}{{{s}_{\text{r},i}}}<{T}_{\text{tgt}}-\delta  \right\} \nonumber  \\
  &\quad\quad\quad\quad\quad +\ldots +\Pr \left\{ {{s}_{\text{r},0}}\ge {T}_{\text{tgt}}-\delta  \right\} \nonumber  \\
  & =1.
\end{align}
Therefore, the probability distribution function of the packet under the condition of ${T}_{\text{tgt}}\ge 1+n_\text{r}+\delta$ is given by
\begin{align}\label{eq:Primary formula of repetition strategy}
   &\Pr \left\{{{S}_{n_\text{r}}}-{T}_{\text{tgt}}>j \right\}\nonumber \\
   &=\left\{ {
   \begin{aligned}
      &1,&& j<n_\text{r}-{T}_{\text{tgt}} \\
      &\sum\limits_{i=0}^{n_\text{r}}C_{{T}_{\text{tgt}}-\delta -1}^{i}{{p}^{i}}{{\left( 1-p \right)}^{j+{T}_{\text{tgt}}-i}},&& j\ge -1-\delta  \\
      &\sum\limits_{i=0}^{n_\text{r}}C_{j+{T}_{\text{tgt}}}^{i}{{p}^{i}}{{\left( 1-p \right)}^{j+{T}_{\text{tgt}}-i}},&& \text{else}.  \\
  \end{aligned}} \right.
\end{align}
\subsubsection{${T}_{\rm{tgt}}\le1+\delta$}
 In case ${T}_{\text{tgt}}\le1+\delta$, the transmission starting time of the packet falls within the target reception range. Thus the packet will not be received before the target reception range and the packet will only be transmitted at most once. We have ${T}_{\text{tgt}}-\delta \le1\le {{s}_{\text{r},0}}\le \sum\limits_{i=0}^{n_\text{r}}{{{s}_{\text{r},i}}}$. From $\eqref{Proof method and formula of probability distribution}$, we then have
\begin{align}
  \Pr \left\{ {{S}_{n_\text{r}}}-{T}_{\text{tgt}}>j \right\}&=\Pr \left\{ {{s}_{\text{r},0}}-{T}_{\text{tgt}}>j \right\} \nonumber  \\
  &=\left\{
  \begin{aligned}
  & {{\left( 1-p \right)}^{j+{T}_{\text{tgt}}}} ,&& j\ge -{T}_{\text{tgt}} \\
  & 1,&& j<-{T}_{\text{tgt}}. \\
  \end{aligned}\right.
\end{align}
\subsubsection{$1+\delta< {T}_{\rm{tgt}}<1+n_\text{r}+\delta$}
We denote $Q={T}_{\text{tgt}}-\delta -1$. Since $n_\text{r}>{T}_{\text{tgt}}-\delta -1$, we can get \eqref{eq:appendix_kualan_3}.
\begin{figure*}
\hrule
\begin{align}\label{eq:appendix_kualan_3}
 & \Pr \left\{ {{S}_{n_\text{r}}}-{T}_{\text{tgt}}>j \right\}\nonumber  \\
 & =\Pr \left\{ \sum\limits_{i=0}^{n_\text{r}}{{s}_{\text{r},i}}-{T}_{\text{tgt}}>j,\sum\limits_{i=0}^{n_\text{r}-1}{{s}_{\text{r},i}}<{T}_{\text{tgt}}-\delta  \right\}+\ldots
   +\Pr \left\{ \sum\limits_{i=0}^{Q}{{s}_{\text{r},i}}-{T}_{\text{tgt}}>j,\sum\limits_{i=0}^{Q}{{s}_{\text{r},i}}\ge {T}_{\text{tgt}}-\delta ,\sum\limits_{i=0}^{Q -1}{{s}_{\text{r},i}}<{T}_{\text{tgt}}-\delta  \right\}\nonumber  \\
 & \qquad +\ldots +\Pr \left\{ {s}_{\text{r},0}-{T}_{\text{tgt}}>j,{s}_{\text{r},0}\ge {T}_{\text{tgt}}-\delta  \right\}\nonumber  \\
 & =\Pr \left\{ \sum\limits_{i=0}^{Q}{{s}_{\text{r},i}}-{T}_{\text{tgt}}>j,\sum\limits_{i=0}^{Q}{{s}_{\text{r},i}}\ge {T}_{\text{tgt}}-\delta ,\sum\limits_{i=0}^{Q -1}{{s}_{\text{r},i}}<{T}_{\text{tgt}}-\delta  \right\}
   +\ldots +\Pr \left\{ {s}_{\text{r},i}-{T}_{\text{tgt}}>j,{s}_{\text{r},i}\ge {T}_{\text{tgt}}-\delta  \right\}\nonumber  \\
 & =\Pr \left\{ \sum\limits_{i=0}^{Q}{{s}_{\text{r},i}}-{T}_{\text{tgt}}>j,\sum\limits_{i=0}^{Q -1}{{s}_{\text{r},i}}<{T}_{\text{tgt}}-\delta  \right\}+
 \Pr \left\{ \sum\limits_{i=0}^{Q-1}{{s}_{\text{r},i}}-{T}_{\text{tgt}}>j,\sum\limits_{i=0}^{Q-1}{{s}_{\text{r},i}}\ge {T}_{\text{tgt}}-\delta ,\sum\limits_{i=0}^{Q -2}{{s}_{\text{r},i}}<{T}_{\text{tgt}}-\delta  \right\}\nonumber  \\
 & \qquad +\ldots +\Pr \left\{ {s}_{\text{r},0}-{T}_{\text{tgt}}>j,{s}_{\text{r},0}\ge {T}_{\text{tgt}}-\delta  \right\}\nonumber  \\
 & =\Pr \left\{ {{S}_{Q}}-{T}_{\text{tgt}}>j \right\}.
\end{align}
\hrule
\end{figure*}
Thus, by substituting $n_\text{r}=Q={{T}_{\text{tgt}}}-1-\delta$ into $\eqref{eq:Primary formula of repetition strategy}$, we can get
\begin{align}
  & \Pr \left\{ {{S}_{n_\text{r}}}-{{T}_{\text{tgt}}}>j \right\}\nonumber \\
  & =\Pr \left\{ {{S}_{{T}_{\text{tgt}}-\delta -1}}-{T}_{\text{tgt}}>j \right\}\nonumber  \\
  & =\left\{
  \begin{aligned}
  & \sum\limits_{m=0}^{{T}_{\text{tgt}}-\delta -1}C_{{T}_{\text{tgt}}-\delta -1}^{m}{{p}^{m}}{{\left( 1-p \right)}^{j+{T}_{\text{tgt}}-m}},&& j\ge -1-\delta \nonumber  \\
  & 1,&& j<-1-\delta  \\
  \end{aligned} \right. \\
    & =\left\{
   \begin{aligned}
  & {{\left( 1-p \right)}^{j+1+\delta }},&& j\ge -1-\delta  \\
  & 1,&& j<-1-\delta.  \\
\end{aligned} \right.
\end{align}

This completes the proof of Proposition \ref{Probability distribution function of repetition strategy}.
\end{proof}

\subsection{Proof of Proposition \ref{eq:Proof formula of reward function}}
\begin{proof}\label{Proof of corollary 3}
Note that in Proposition \ref{Probability distribution function of repetition strategy}, the system transmits only one packet and the state of the packet is $i={T}_{\text{tgt}}$. Using the results in Proposition \ref{Probability distribution function of repetition strategy} and substituting $y=j+{T}_{\text{tgt}}$ and $i={T}_{\text{tgt}}$ into \eqref{corollary2.1}, \eqref{corollary2.2} and \eqref{corollary2.3}, for the case  $i\ge 1+n_\text{r}+\delta$ we have
\begin{align}\label{eq:Repeated variant 1}
   &\Pr \left\{{{S}_{n_\text{r}}}>y \right\}\nonumber \\
   &=\left\{ {
   \begin{aligned}
      &\sum\limits_{m=0}^{n_\text{r}}C_{y}^{m}{{p}^{m}}{{\left( 1-p \right)}^{y-m}},&& n_\text{r}\le y<i-1-\delta  \\
      &\sum\limits_{m=0}^{n_\text{r}}C_{i-\delta -1}^{m}{{p}^{m}}{{\left( 1-p \right)}^{y-m}},&& y\ge i-1-\delta  \\
      &1,&& y<n_\text{r}; \\
  \end{aligned}} \right.
\end{align}
when $i \le 1+\delta$, we have
\begin{align}\label{eq:Repeated variant 2}
  \Pr \left\{ {{S}_{n_\text{r}}}>y \right\}& =\Pr \left\{ {{S}_{0}}>y \right\}=\left\{
  \begin{aligned}
  & {{\left( 1-p \right)}^{y}},&& y\ge0 \\
  & 1,&& y<0; \\
\end{aligned} \right.
\end{align}
when $1+\delta< i<1+n_\text{r}+\delta$, we have
\begin{align}\label{eq:Repeated variant 3}
  \Pr \left\{ {{S}_{n_\text{r}}}>y \right\}& =\Pr \left\{ {{S}_{i-1-\delta }}>y \right\}\nonumber \\
  &=\left\{
  \begin{aligned}
  & {{\left( 1-p \right)}^{y-i+1+\delta }},&& y\ge i-1-\delta  \\
  & 1,&& y<i-1-\delta.  \\
  \end{aligned} \right.
\end{align}
We denote the transition probability of the packet from state $i$ to state $j$ by adopting the repeat strategy as ${{p}_{ij}}\left( n_\text{r} \right)$. From $\eqref{eq:transfer}$, we have
\begin{align}\label{eq:Repeated solution}
  {{p}_{ij}}\left( n_\text{r} \right)& =\Pr \left\{ i-j+{T}_{\text{tgt}}=y \right\}\nonumber  \\
  & =\Pr \left\{ {{S}_{n_\text{r}}}=y \right\}\nonumber  \\
  & =\Pr \left\{ {{S}_{n_\text{r}}}>y-1 \right\}-\Pr \left\{ {{S}_{n_\text{r}}}>y \right\}.
\end{align}
\subsubsection{$i \le 1+\delta$}
When $i \le 1+\delta$, from $\eqref{eq:Repeated variant 2}$ and $\eqref{eq:Repeated solution}$, we have
\begin{align}
  &{{p}_{ij}}\left( n_\text{r} \right) \nonumber  \\
  &=\Pr \left\{ i-j+{T}_{\text{tgt}}=y \right\} \nonumber  \\
  & =\Pr \left\{ {{S}_{n_\text{r}}}=y \right\} \nonumber  \\
  & =\Pr \left\{ {{S}_{n_\text{r}}}>y-1 \right\}-\Pr \left\{ {{S}_{n_\text{r}}}>y \right\} \nonumber  \\
  & =\Pr \left\{ {{S}_{n_\text{r}}}>i-j+{T}_{\text{tgt}}-1 \right\}-\Pr \left\{ {{S}_{n_\text{r}}}>i-j+{T}_{\text{tgt}} \right\} \nonumber  \\
  & =\left\{
  \begin{aligned}
  & {{\left( 1-p \right)}^{i-j+{T}_{\text{tgt}}-1}}-{{\left( 1-p \right)}^{i-j+{T}_{\text{tgt}}}},&& i-j+{T}_{\text{tgt}}>0 \nonumber \\
  & 1-1,&& i-j+{T}_{\text{tgt}}\le 0 \\
  \end{aligned} \right.\\
  & =\left\{
  \begin{aligned}
  & p{{\left( 1-p \right)}^{i-j+{T}_{\text{tgt}}-1}},&& i>j-{T}_{\text{tgt}} \\
  & 0,&& i\le j-{T}_{\text{tgt}}. \\
  \end{aligned} \right.
\end{align}
\subsubsection{$1+\delta< i<1+n_\text{r}+\delta$}
When $1+\delta< i<1+n_\text{r}+\delta$, from $\eqref{eq:Repeated variant 3}$ and $\eqref{eq:Repeated solution}$, we have
\begin{align}
  & {{p}_{ij}}\left( n_\text{r} \right) \nonumber  \\
  & =\Pr \left\{ i-j+{T}_{\text{tgt}}=y \right\} \nonumber  \\
  & =\Pr \left\{ {{S}_{n_\text{r}}}=y \right\} \nonumber  \\
  & =\Pr \left\{ {{S}_{n_\text{r}}}>y-1 \right\}-\Pr \left\{ {{S}_{n_\text{r}}}>y \right\} \nonumber  \\
  & =\Pr \left\{ {{S}_{n_\text{r}}}>i-j+{T}_{\text{tgt}}-1 \right\}-\Pr \left\{ {{S}_{n_\text{r}}}>i-j+{T}_{\text{tgt}} \right\} \nonumber  \\
  & =\left\{
   \begin{aligned}
  & p{{\left( 1-p \right)}^{{T}_{\text{tgt}}-j+\delta }},&& j<1+\delta +{T}_{\text{tgt}} \\
  & 0,&& j\ge 1+\delta +{T}_{\text{tgt}}. \\
   \end{aligned} \right.
\end{align}
\subsubsection{$i\ge 1+n_\text{r}+\delta$}
In the case of $i\ge 1+n_\text{r}+\delta$, we first consider the state transition probability when $n_\text{r}< y\le i-1-\delta$. From $\eqref{eq:Repeated variant 1}$ and $\eqref{eq:Repeated solution}$, we have
\begin{align}
  & {{p}_{ij}}\left( n_\text{r} \right) \nonumber \\
  & =\Pr \left\{ i-j+{T}_{\text{tgt}}=y \right\} \nonumber  \\
  & =\Pr \left\{ i-j+{T}_{\text{tgt}}>y-1 \right\}-\Pr \left\{ i-j+{T}_{\text{tgt}}>y \right\} \nonumber  \\
  & =\sum\limits_{m=0}^{n_\text{r}}{C_{y-1}^{m}{{p}^{m}}{{\left( 1-p \right)}^{y-1-m}}}-\sum\limits_{m=0}^{n_\text{r}}{C_{y}^{m}{{p}^{m}}{{\left( 1-p \right)}^{y-m}}} \nonumber  \\
  & =\sum\limits_{m=0}^{n_\text{r}}{C_{y-1}^{m}{{p}^{m}}{{\left( 1-p \right)}^{y-1-m}}} \nonumber \\
  &\quad\quad\quad\quad\quad\quad\ -\left( 1-p \right)\sum\limits_{m=0}^{n_\text{r}}{\frac{y}{y-m}C_{y-1}^{m}{{p}^{m}}{{\left( 1-p \right)}^{y-1-m}}} \nonumber  \\
  & =p\sum\limits_{m=0}^{n_\text{r}}{C_{y-1}^{m}{{p}^{m}}{{\left( 1-p \right)}^{y-1-m}}}\nonumber \\
  &\quad\quad\quad\quad\quad\quad\ -\left( 1-p \right)\sum\limits_{m=0}^{n_\text{r}}{\frac{m}{y-m}C_{y}^{m}{{p}^{m}}{{\left( 1-p \right)}^{y-1-m}}} \nonumber  \\
  & =\sum\limits_{m=0}^{n_\text{r}}{C_{y-1}^{m}{{p}^{m+1}}{{\left( 1-p \right)}^{y-1-m}}} \nonumber\\
  &\quad\quad\quad\quad\quad\quad\ -\sum\limits_{m=1}^{n_\text{r}}{C_{y-1}^{m-1}{{p}^{m}}{{\left( 1-p \right)}^{y-m}}} \nonumber  \\
  & =\sum\limits_{m=0}^{n_\text{r}}{C_{y-1}^{m}{{p}^{m+1}}{{\left( 1-p \right)}^{y-1-m}}} \nonumber \\
  &\quad\quad\quad\quad\quad\quad\ -\sum\limits_{m=0}^{n_\text{r}-1}{C_{y-1}^{m}{{p}^{m+1}}{{\left( 1-p \right)}^{y-1-m}}}& \nonumber \\
  & =C_{y-1}^{n_\text{r}}{{p}^{1+n_\text{r}}}{{\left( 1-p \right)}^{y-1-n_\text{r}}}.
\end{align}
We then consider the state transition probability when $y> i-1-\delta$ and have
\begin{align}
  & {{p}_{ij}}\left( n_\text{r} \right) \nonumber \\
  & =\Pr \left\{ i-j+{T}_{\text{tgt}}=y \right\}\nonumber  \\
  & =\Pr \left\{ i-j+{T}_{\text{tgt}}>y-1 \right\}-\Pr \left\{ i-j+{T}_{\text{tgt}}>y \right\}\nonumber  \\
  & =\sum\limits_{m=0}^{n_\text{r}}{C_{i-1-\delta }^{m}{{p}^{m}}{{\left( 1-p \right)}^{y-m-1}}}\nonumber  \\
  &\qquad\qquad\qquad\qquad\qquad\qquad -\sum\limits_{m=0}^{n_\text{r}}{C_{i-1-\delta }^{m}{{p}^{m}}{{\left( 1-p \right)}^{y-m}}}\nonumber  \\
  & =\sum\limits_{m=0}^{n_\text{r}}{C_{i-1-\delta }^{m}{{p}^{1+m}}{{\left( 1-p \right)}^{y-m-1}}}.
\end{align}
Finally, when $y\le k$, we have
\begin{align}
  & {{p}_{ij}}\left( n_\text{r} \right) \nonumber \\
  & =\Pr \left\{ i-j+{T}_{\text{tgt}}=y \right\} \nonumber \\
  & =\Pr \left\{ i-j+{T}_{\text{tgt}}>y-1 \right\}-\Pr \left\{ i-j+{T}_{\text{tgt}}>y \right\} \nonumber  \\
  & =1-1 \nonumber  \\
  & =0,
\end{align}
Thus when $i\ge 1+n_\text{r}+\delta$, we have
\begin{align}
 & {{p}_{ij}}\left( n_\text{r} \right) \nonumber \\
 & =\left\{
 \begin{aligned}
 & 0&& ,y\le n_\text{r} \\
 & C_{y-1}^{n_\text{r}}{{p}^{1+n_\text{r}}}{{\left( 1-p \right)}^{y-1-n_\text{r}}}&& ,n_\text{r}<y\le i-1-\delta \\
 & \sum\limits_{m=0}^{n_\text{r}}{C_{i-1-\delta }^{m}{{p}^{1+m}}{{\left( 1-p \right)}^{^{y-1-m}}}}&& ,y> i-1-\delta,
\end{aligned} \right.
\end{align}
in which $y=i-j+{T}_{\text{tgt}}$.

This completes the proof of Proposition \ref{eq:Proof formula of reward function}.
\end{proof}

\end{document}